\pdfoutput=1

\documentclass[format=acmsmall, review=false, screen=true]{acmart}

\usepackage{amssymb}
\usepackage{xspace}
\usepackage{tikz}
\usepackage{xcolor}
\usepackage{mathbbol}
\usepackage{todonotes}
\usepackage{graphicx}
\usepackage[font={small}]{caption,subcaption}
\usepackage[vlined,ruled,linesnumbered]{algorithm2e}
\hypersetup{
	bookmarks=false
}

\newcommand{\ignore}[1]{\relax}

\newcommand{\set}[1]{\{ #1 \}}

\newcommand{\cL}{\mathcal{L}}

\newcommand{\Ostar}[1]{{O^\star}\!\left( {#1} \right)}
\newcommand{\Oo}[1]{#1} 

\newcommand{\Ostart}[1]{{O^\star}( {#1} )}

\newcommand{\Thstar}[1]{{\Theta^\star}\left( {#1} \right)}

\renewcommand{\sl}[1][]{ |L_3#1| }
\renewcommand{\ss}[1][]{ |S_3#1| }
\newcommand{\sr}[1][]{ |R_3#1| }

\newcommand{\sst}[1][]{ |S_2#1| }
\newcommand{\sv}[1][]{ |V#1| }
\renewcommand{\wr}{w_r}

\newcommand{\ws}{w_s}
\newcommand{\wst}{w_{s,2}}

\newcommand{\wb}{w_b}
\newcommand{\wc}{w_c}
\newcommand{\wdd}{w_d}
\newcommand{\one}{ \mathbb{1} }
\def\one(#1){ \mathbb{1} \! \left(#1\right) }
\newcommand{\mytab}{\hspace{0.5cm}}
\newcommand{\e}{\epsilon}
\newcommand{\redsec}[1]{\medskip\paragraph{\textbf{#1}}}

\newcommand{\nds}{\textsc{\#Dominating Set}\xspace}
\newcommand{\DS}{\textsc{DS}\xspace}
\newcommand{\NDS}{\textsc{\#DS}\xspace}
\newcommand{\mcsp}{\textsc{Max 2-CSP}\xspace}
\newcommand{\mrcsp}{\textsc{Max $(r,2)$-CSP}\xspace}
\newcommand{\mtcsp}{\textsc{Max $(2,2)$-CSP}\xspace}
\newcommand{\card}[1]{\left| #1 \right|}
\newcommand{\ceil}[1]{\left\lceil #1 \right\rceil}
\newcommand{\mc}{Measure and Conquer\xspace}
\newcommand{\smc}{Separate, Measure and Conquer\xspace}
\newcommand{\rt}[2]{{#1}^{#2}}

\newcommand{\Red}[1]{Reduction #1\xspace}	
\newcommand{\Reds}[1]{Reductions #1\xspace}

\SetAlFnt{\small}
\SetAlCapFnt{\small}
\SetAlCapNameFnt{\small}
\SetAlCapHSkip{0pt}
\IncMargin{-\parindent}

\setcopyright{none}
\acmYear{} 
\acmDOI{} 

\fancypagestyle{firstpagestyle}{%
	\fancyfoot[RO,LE]{}
	\fancyhead[RO,RE,LO,LE]{}
}
\renewcommand\footnotetextcopyrightpermission[1]{}  

\begin{document}

\title[Separate, Measure and Conquer]%
 {Separate, Measure and Conquer:
	 Faster Polynomial-Space Algorithms
	 for Max 2-CSP and Counting Dominating Sets}

\author{Serge Gaspers}
\orcid{0000-0002-6947-9238}
\affiliation{%
	\institution{UNSW Sydney}
	\city{Sydney}
	\state{NSW}
	\postcode{2052}
	\country{Australia}}
\affiliation{%
	\institution{Data61, CSIRO}
	\city{Sydney}
	\state{NSW}
	\postcode{2052}
	\country{Australia}}

\author{Gregory B. Sorkin}
\affiliation{%
 \institution{London School of Economics}
 \streetaddress{Houghton Street}
 \city{London}
 \postcode{WC2A 2AE}
 \country{UK}}

\tikzstyle{every node}=[minimum size=2mm]
\tikzstyle{vertex}=[minimum size=2mm,circle,fill=black,inner sep=0mm,draw]
\tikzstyle{mydashed}=[dash pattern=on 2pt off 3pt]

\citestyle{acmnumeric}
\setcitestyle{numbers,sort&compress}

\begin{abstract}
		We show a method resulting in the improvement of
		several polynomial-space, ex\-po\-nen\-tial-time algorithms.
		The method capitalizes on the existence of small balanced separators for
		sparse graphs, which can be exploited for branching to disconnect
		an instance into independent components.
		For this algorithm design paradigm, the challenge to date has been to obtain improvements in
		worst-case analyses of algorithms, compared with algorithms that are analyzed
		with advanced methods, notably \mc.
		Our contribution is the design of a general method to integrate the advantage from
		the separator-branching into \mc, for a more precise and improved running time analysis.
		
		We illustrate the method with improved algorithms for \mrcsp and \nds.
An instance of the problem \mrcsp, 
or simply \mcsp, 
is parameterized by the domain size $r$ (often~2),
the number of variables $n$ (vertices in the constraint graph $G$),
and the number of constraints $m$ (edges in $G$).
When $G$ is cubic,
and omitting sub-exponential terms here for clarity,
we give an algorithm running in time
$\rt r {(1/5) \, n} = \rt r {(2/15) \, m}$;
the previous best was
$\rt r {(1/4) \, n} = \rt r {(1/6) \, m}$.
By known results, this improvement for the cubic case results in
an algorithm running in time
$\rt r {(9/50) \, m}$
for general instances;
the previous best was
$\rt r {(19/100) \, m}$.
We show that the analysis of the earlier algorithm was tight:
our improvement is in the algorithm, not just the analysis.
The same running time improvements hold for \textsc{Max Cut}, an important
special case of \mcsp,
and for Polynomial and Ring CSP,
generalizations encompassing graph bisection, the Ising model,
and counting.

We also give faster algorithms for \nds, 
counting the dominating sets of every cardinality $0,\ldots,n$
for a graph $G$ of order $n$.
For cubic graphs, our algorithm runs in time
$\rt 3 {(1/5) \, n}$;
the previous best was
$\rt 2 {(1/2) \, n}$.
For general graphs, we give an unrelated algorithm running in time
$\rt {1.5183} n$;
the previous best was
$\rt {1.5673} n$.

The previous best algorithms for these problems all used local
transformations and were analyzed by the \mc method.
Our new algorithms capitalize on the existence of small balanced
separators for cubic graphs --- a non-local property ---
and the ability to tailor the local algorithms
always to ``pivot'' on a vertex in the separator.
The new algorithms perform much as the old ones until the separator is empty,
at which point they gain because
the remaining vertices are split into two independent
problem instances that can be solved recursively.
It is likely that such algorithms can be effective for other problems too,
and we present their design and analysis in a general framework.
\end{abstract}

\ccsdesc[500]{Theory of computation~Graph algorithms analysis}
\ccsdesc[500]{Theory of computation~Parameterized complexity and exact algorithms}
\ccsdesc[300]{Theory of computation~Backtracking}

\keywords{Measure \& Conquer, graph separators, Max 2-CSP, Counting dominating sets}

\thanks{
	Authors' addresses: S.~Gaspers, School of Computer Science and Engineering,
	Building K17, UNSW Sydney, Sydney NSW 2052, Australia;
	G.~B.~Sorkin, Department of Mathematics, London School of Economics, Houghton Street, London WC2A 2AE, UK.
}

\maketitle

\section{Introduction}

\subsection{Background and intuition of the method}
\label{sec:intuition}

\mc \cite{FominGK09} has become the prevalent method to analyze branching algorithms over the last few years.
Several textbooks \cite{FominK10,DowneyF13} and PhD theses \cite{BinkeleRaible10,Cochefert14,Couturier12,Gaspers10,Grandoni04,Liedloff07,Rooij11,Stepanov08,Wahlstrom07} give an extensive treatment of the method and examples of algorithms analyzed with \mc.
Many of these algorithms first branch on vertices of high degree before handling graphs with bounded maximum degree.
Improvements in running times are often due to faster computations on small-degree graphs, and the \mc method then allows to propagate this speed-up to higher degrees, for an overall improved running time.

Some general methods to obtain improvements in the running time lead to exponential-space algorithms.
Among them are methods based on memoization (see, e.g., \cite{Robson86,Grandoni06}) and dynamic programming on tree decompositions or path decompositions (see, e.g., \cite{FominH06,FominGSS09}).
In this paper, we aim to obtain comparable improvements as when performing dynamic programming on tree decompositions or path decompositions, but without using exponential space.
Our method can be seen as exploiting the separation properties of path/tree decompositions in a branching process instead of dynamic programming.
The design of fast exponential-time algorithms with restricted space usage has attracted attention for many problems, such as Steiner Tree \cite{Nederlof13,FominGKLS13}, Knapsack \cite{LokshtanovN10}, Hamiltonian Path \cite{Karp82}, and Coloring \cite{BjorklundHK09}. Further results and open problems surrounding the trade-off between time and space complexities are discussed in a survey by Woeginger \cite{Woeginger04}.

The use of graph separators for divide-and-conquer algorithms dates back to the 1970s \cite{LiptonT77}.
For classes of instances with sublinear separators,
including those described by
planar graphs, this
often gives subexponential- or polynomial-time algorithms.
%
It is natural to design a branching strategy that strives to disconnect an instance into components, even when no sublinear separators are known.
While this has successfully been done experimentally (see, e.g., \cite{BiereS06,DechterM07,FreuderQ85,GottlobLS00,LiB04}), we are not aware of worst-case analyses of branching algorithms that are based on linear sized separators.\footnote{Fomin et al. \cite{FominGLS12} exploit small balanced structures of a solution that is to be found in a graph, whereas our separators are separators for the input graph.}%

Our algorithms exploit small separators,
specifically, balanced separators of size about $n/6$
for cubic graphs of order $n$.
The existence of such separators has been known since 2001 at least,
they have been used in pathwidth-based dynamic programming algorithms
using exponential time and exponential space,
and
it is natural to try to exploit them
for algorithms running in exponential time but polynomial space.

\medskip

\noindent
\textbf{Outline of the paper.}
(A more detailed outline follows as Section~\ref{sec:outline}.)
In the rest of the Introduction we introduce the separation results and discuss applying them to \mrcsp, 
or simply \mcsp.
In \mcsp, the input is a graph, together with a score function for each vertex and each edge whose values depend on the color of the vertex or the endpoints of the edge, respectively. The task is to color the vertices with colors from $[r]=\{1,\dots,r\}$ such that the sum of all scores is maximized. An important special case is \textsc{Max Cut} where $r=2$ and the score functions always return 0, except when it is a score function of an edge and the two endpoints of the edge do not have the same color, in which case the score function returns 1.
We will show how applying the separation results naively, using the separator alone, gives an algorithm worse than existing ones
for \mcsp%
; perhaps this is why such algorithms have not already been developed.
Then, we sketch how, under optimistic assumptions,
using small separators \emph{in conjunction with}
existing \mc analyses gives an improvement. 
In the body of the paper
we turn this into a rigorous analysis
of a faster algorithm for \mcsp.
We then describe a general framework,
``\smc'', to design and analyze
polynomial-space, exponential-time algorithms based on
small balanced separators.
We use it to obtain faster algorithms for \nds,
calling upon the method in significantly greater generality
than was needed for \mcsp.
We conclude with some thoughts on the method,
including the likely limitation that the (polynomial-space)
algorithms based on it
will be slower than exponential-space dynamic programming algorithms
based on path decompositions.

\medskip

We now introduce the main separation results we will use.
The \emph{bisection width} of a graph is the minimum number of edges between
vertices belonging to different parts,
over all partitions of the vertex sets into two parts whose size differs by at most one.
Monien and Preis~\cite{MonienP01,MonienP06} proved
the following upper bound on the bisection width of cubic graphs.

\begin{theorem}[\cite{MonienP06}]
For any $\epsilon > 0$ there is a value $n_\epsilon$ such that the bisection width of any cubic graph $G = (V ,E)$ with $|V |>n_\epsilon$ is at most $(\frac{1}{6} + \epsilon) |V |$.
\end{theorem}

\noindent
As noted in \cite{FominH06}, the bound also holds for graphs of maximum degree at most 3 and a corresponding bisection can be computed in polynomial time.
Kostochka and Melnikov \cite{KostochkaM93} showed that almost all cubic graphs on $n$ vertices have bisection width at least $n/9.9$. To the best of our knowledge, this is the best known lower bound on the bisection width of cubic graphs.

Let $(L,S,R)$ be a partition of the vertex set of a graph $G$ such that
there is no edge in $G$ with one endpoint in $L$ and the other endpoint in
$R$. We say that $(L,S,R)$ is a \emph{separation} of $G$, and that $S$ is a
\emph{separator} of $G$, \emph{separating} $L$ and $R$ (thought of as Left and Right).
The previous theorem immediately gives a small \emph{balanced} separator, one partitioning the remaining vertices into equal-sized parts.
It is obtained by choosing a vertex cover of the edges in the bisection.

\begin{lemma} \label{lem:sep-cubic}
For any $\epsilon > 0$ there is a value $n_\epsilon$ such that every graph $G=(V,E)$ of maximum degree at most 3,
$\card V \geq n_\epsilon$,
has a separation $(L,S,R)$ with $\card S \leq (\frac{1}{6} + \epsilon) |V |$
and $|L|,|R| \le \ceil{ \frac{|V|-|S|}{2} }$. Moreover, such a separation can be computed in polynomial time.
\end{lemma}

\smallskip
\noindent
A direct application of branching on the vertices in such a separator
yields algorithms inferior to existing \mc\ ones.
We illustrate with solving a cubic \mcsp instance with $n$ vertices.
We separate the vertices as $(L,S,R)$ per Lemma~\ref{lem:sep-cubic}.
Sequentially, we generate $r^{\card S}$
smaller instances by taking each possible assignment to
the vertices of $S$.
For \mcsp, as for many problems,
assigning a value to a vertex $v \in V$
yields an instance on vertices $V \setminus \set{v}$,
so the procedure above produces instances on the vertex set $L \cup R$.
The restriction of $G$ to those vertices,
$G \left|_{L \cup R}\right.$,
is by construction a graph with no edge between $L$ and $R$,
and the solution to an instance on constraint graph
$G \left|_{L \cup R}\right.$
is simply the direct combination of the solutions to
the corresponding instances given by
$G \left|_{L} \right.$
and $G \left|_{R} \right.$.
The advantage given by reducing on vertices in the \emph{separator}
is that (for each branch)
we must solve two small instances rather than one large one.
Roughly speaking, in the worst case where $\card S = \tfrac16 n$
and $\card L = \card R = \tfrac5{12}n$,
the running time $t(n)$ satisfies the recurrence
\begin{align*}
t(n) &= r^{n/6} (2 \cdot t(\tfrac5{12} n)) ,
\end{align*}
whose solution satisfies $t(n) = \Ostart{r^{2n/7}}$.%
\footnote{The notation $f(n)=\Ostart{g(n)}$
indicates that for some polynomial $p(n)$,
for all sufficiently large $n$, $f(n) \leq p(n) g(n)$.}
This is inferior to the $\Ostart{r^{n/4}}$ upper bound from
\cite{linear,ScottS07}
for a simple polynomial-space algorithm
based on local simplification and branching rules.

The improvements in this paper
have their origin in a simple observation:
if an algorithm can always branch
on vertices in the separator, then
the usual measure of improvement is achieved at each step
(typically computed by the \mc method described
in Subsection~\ref{subsec:mc}),
and the splitting of the graph into two parts when the separator is emptied
is a bonus. We get the best of both.
The technical challenges are 
to accurately amortize this bonus over the previous branches
to prove a better running time,
and to control the balance of the separation as the algorithm proceeds
so that the bonus is significant.

Again, we illustrate our approach for cubic \mcsp.
As remarked earlier, we will be optimistic in this sketch,
doing the analysis rigorously in Section~\ref{sec:csp}.
The problem class will also be defined there, but for now one may think
of \textsc{Max Cut}, with domain size $r=2$.
Let us ``pivot'' on a vertex $v \in S$, i.e.,
sequentially assign it each possible value,
eliminate it and its incident edges
(see rule \Red{3} below),
and solve each case recursively.
It is possible that $v$ has one or more neighbors within $S$,
but this is a favorable case,
reducing the number of subsequent branches needed.
So, suppose that $v$ has neighbors only in $L$ and $R$.
If all its neighbors were in one part, the separator could be made smaller,
so let us skip over this case as well.
The cases of interest, then,
are when $v$ has two neighbors in $L$ and one in $R$, or vice-versa.
Suppose that these cases occur equally often;
this is the bit of optimism that will require more care
to get right.
In that case, after all $\card S$ branchings,
the sizes of $L$ and $R$ are each reduced by $\tfrac32 \card S$, since degree-2 vertices get contracted away.
This would lead to a running time bound $t(n)$ satisfying the recurrence
\begin{align*}
t(n) &= r^{n/6} \cdot 2 t(\tfrac5{12}n-\tfrac32 \cdot \tfrac16 n) ,
\end{align*}
leading to a solution with
$t(n)= \Ostart{r^{n/5}}$.
This conjectured bound would improve on the best previous algorithm's
time bound of
$\Ostart{r^{n/4}}$,
and Section~\ref{sec:csp} establishes that the bound is true,
modulo a subexponential factor in the running time
due to Lemma~\ref{lem:sep-cubic}
guaranteeing only $\card S \leq (\frac{1}{6} + \epsilon) n$ instead of $\card S \leq \frac{1}{6} n$.

Our algorithms exploit a \emph{global} graph structure,
the separator, while executing an algorithm based on \emph{local}
simplification and branching rules.
The use of global structure may also make it possible to circumvent
lower bounds for classes of branching algorithms that only consider local information when deciding whether to branch on a variable \cite{AchSo00,AlekhnovichHI05}.

\subsection{Results and organization of the paper} \label{sec:outline}
Section~\ref{sec:prelim} defines notation and gives the necessary background on separators and the \mc method that are necessary for Section~\ref{sec:csp}.
Section~\ref{sec:csp} gives a first analysis of a separator-based algorithm.
It solves cubic instances of \mcsp in time $\Oo{r^{(1/5+o(1)) \, n}}$, where
$n$ is the number of vertices.
This improves on
the previously fastest $\Ostar{r^{n/4}}$ time
polynomial-space algorithm
\cite{ScottS07}. The fastest known exponential-space algorithm for cubic instances
uses the pathwidth approach of Fomin and H{\o}ie \cite{FominH06} to solve \mcsp in time $\Oo{r^{(1/6 + o(1)) \, n}}$ \cite{ScottS07}.

Our algorithm also improves the fastest known running time for polynomial-space algorithms solving general instances
of \mcsp from
$\Ostar{r^{(19/100) \, m}}$ to $\Oo{r^{(9/50 + o(1)) \, m}}$, where $m$ is the number of edges of the constraint graph.
The same running time improvement holds for \textsc{Max Cut}, an important
special case of \mcsp,
and for Polynomial and Ring CSP,
generalizations encompassing graph bisection, the Ising model,
and counting problems. 
For comparison, the fastest known exponential-space algorithm has worst-case running time $\Oo{r^{(13/75+o(1)) \, m}}$ \cite{ScottS07}. For vertex-parameterized running times, Williams' \cite{Williams05} exponential-space algorithm runs in time $\Ostar(2^{(\omega/3) \, n})$, which is $O(1.7303^n)$ using the current best upper bound on the matrix-multiplication exponent $\omega$ \cite{Gall14a}.

Subsection~\ref{subsec:lb} gives tight lower bounds for the analysis of the previously fastest polynomial-space \mcsp algorithm.
These lower bounds highlight that our new algorithm is strictly faster. %
It should be noted that tight analyses are rare for competitive branching algorithms. An important feature of these lower bounds is that they make the algorithm use several branching rules whose \mc analyses are most constraining.

In independent work, Edwards \cite{Edwards16} matches our running times for
\mcsp on general and cubic graphs.
His work is also based on separators and is similar to ours from an algorithmic point of view.
The running time analysis differs from ours in that it is targeted solely at \mcsp, while we provide a general method, as we discuss next.

While \mcsp is a central problem in exponential-time algorithms,
the analysis of its branching algorithms is typically easier than for other problems,
largely because the branching creates isomorphic subinstances.
In Section~\ref{sec:smc}, we develop the \smc method in full generality,
and use this in Section~\ref{sec:cds} to design faster polynomial-space algorithms for counting dominating sets in graphs.
For graphs with maximum degree $3$, we obtain an algorithm with a time bound of
$\Oo{3^{(1/5+o(1)) \, n}} = O(1.2458^n)$,
improving on the previous best polynomial-space running time of
$\Ostart{2^{(1/2) \, n}} = O(1.4143^n)$ \cite{KneisMRR05}.
For general graphs, we obtain a different algorithm,
with time bound $O(1.5183^n)$,
improving on the previous best polynomial-space running time of $O(1.5673^n)$ \cite{Rooij10}.
For comparison, the fastest known exponential-space algorithms run in time $\Oo{3^{(1/6+o(1)) \, n}}$ for cubic graphs (by combining \cite{RooijBR09} and \cite{FominH06}) and in time $O(1.5002^n)$ for general graphs \cite{NederlofRD14}.

Limitations of the \smc method are discussed in Section \ref{sec:limitations}, and we provide an outlook in the conclusion section.

The hasty reader may choose to skip Subsections \ref{subsec:lb}, \ref{subsec:3cds}, and/or \ref{subsec:cds}, or Section~\ref{sec:cds} altogether.
A preliminary version of this article appeared in the proceeding of ICALP 2015 \cite{GaspersS15}.

\section{Preliminaries}
\label{sec:prelim}

\subsection{Graphs}

Let $G=(V,E)$ be a (simple, undirected) graph, $v\in V$ be a vertex and $S\subseteq V$ a vertex subset of $G$.
We also refer to the vertex set and edge set of $G$ by $V(G)$ and $E(G)$, respectively.
We denote the \emph{open} and \emph{closed neighborhoods} of $v$ by $N_G(v) = \set{u\in V : uv\in E}$ and $N_G[v] = N_G(v)\cup \set{v}$, respectively.
The \emph{closed neighborhood} of $S$ is $N_G[S] = \bigcup_{v\in S} N_G[v]$ and its \emph{open neighborhood} is $N_G(S) = N_G[S] \setminus S$.
The set $S$ and vertex $v$ are said to \emph{dominate} the vertices in $N_G[S]$ and $N_G[v]$, respectively.
The set $S$ is a \emph{dominating set} of $G$ if $N_G[S]=V$.
The degree of $v$ in $G$ is $d_G(v) = |N_G(v)|$.
We sometimes omit the subscripts when they are clear from the context.
The \emph{contraction} of an edge $uv \in E$ is an operation replacing
$u$ and $v$ by one new vertex $c_{uv}$ that is adjacent to $N_G(\set{u,v})$.
The graph $G-S$ is obtained from $G$ by removing the vertices in $S$ and all incident edges. When $S = \set{v}$ we also write $G - v$ instead of $G - \set{v}$.
The subgraph induced on $S$ is $G[S] = G - (V\setminus S)$.
The graph $G$ is \emph{cubic} or \emph{$3$-regular} if each vertex has degree $3$ and \emph{subcubic} if each vertex has degree at most $3$.

A \emph{tree decomposition} of $G$ is a pair $(\{X_i : i\in I\},T)$
where each so-called \emph{bag} $X_i \subseteq V$, $i\in I$, and $T$ is a tree with elements of $I$ as nodes such that:
\begin{enumerate}
  \item for each edge $uv\in E$, there is an $i\in I$ such that $\{u,v\} \subseteq X_i$, and
  \item for each vertex $v\in V$, $T[\set{i\in I: v\in X_i}]$ is a (connected) tree with at least one node.%
\end{enumerate}
The \emph{width} of a tree decomposition is $\max_{i \in I} |X_i|-1$.
The \emph{treewidth} \cite{RobertsonS86} of $G$
is the minimum width taken over all tree decompositions
of $G$.
\emph{Path decompositions} and \emph{pathwidth} are defined similarly, except that $T$ is restricted to be a path.

\subsection{Separators}
\label{subsec:sep}

%
%
%
%

Our algorithms will typically pivot (or branch) on a vertex of the separator $S$, separating $L$ and $R$, producing instances with smaller separators.
As reductions might render $L$ much smaller than $R$, we would sometimes like to rebalance $L$ and $R$ without increasing the size of the separator.

\begin{lemma}
Let $G=(V,E)$ be a subcubic graph and let $S$ be a separator of $G$ separating $L$ and $R$. If there is a vertex $s\in S$ with exactly one neighbor $l$ in $L$, then there is a separator $S'$ of $G$ separating $L'$ and $R'$ such that $|S'| = |S|$ and $|L'|=|L|-1$.
\end{lemma}
\begin{proof}
We use the partition $(L',S',R') = (L\setminus \{l\},(S\setminus \{s\})\cup \{l\}, R\cup\{s\})$.
\end{proof}
We say in this case that we \emph{drag} $s$ into $R$.
This transformation, though it does not change the ``true'' problem instance,
does change its presentation and its measure;
it is treated like any other transformation in the \mc analysis,
to which we turn next.

\subsection{\mc}
\label{subsec:mc}

In this subsection, we recall the basics of the \mc method.
The method was introduced with that name by \cite{FominGK09}
but closely parallels
Eppstein's quasi-convex method \cite{Eppstein06} 
and Scott and Sorkin's linear-programming approach \cite{ScottS07}.
We use the method in a form developed in \cite{hybrid},
yielding convex mathematical programs.

To track the progress a branching algorithm makes when solving an instance, a \mc analysis assigns a potential function to instances, a so-called measure.

\begin{definition}\index{measure}
A \emph{measure} $\mu$ for a problem $P$ is a function from the set of all instances for $P$ to the set of
non negative reals.
\end{definition}

A \mc analysis upper bounds the running time of a branching algorithm by analyzing the size of its search trees,
which model the recursive calls the algorithm makes during an execution.
Two values are of interest, an upper bound on the depth of the search trees and an upper bound on their number of leaves.
In the following lemma, the measure $\eta$ is typically polynomially bounded and used to bound the depth, and $\mu$ is used to bound the number of leaves.
It is proved by a simple induction.

\begin{lemma}[\cite{Gaspers10,hybrid}] \label{lem:measureanalysis}
Let $A$ be an algorithm for a problem $P$,
$c \ge 0$ and $r>1$ be constants, and $\mu(\cdot), \eta(\cdot)$ be measures
for the instances of $P$,
such that for any input instance $I$,
$A$ reduces $I$ to instances $I_1,\ldots,I_k$, with $k\ge 1$,
solves these recursively, and combines their solutions to solve~$I$,
using time $O(\eta(I)^{c})$ for the reduction and combination steps
(but not the recursive solves), with
\begin{align}
(\forall i) \quad \eta(I_i) & \leq \eta(I)-1 \text{, and}  \label{eq:masize}
  \\
\sum_{i=1}^k r^{\mu(I_i)} & \leq r^{\mu(I)} . \label{eq:magain}
\end{align}
Then $A$ solves any instance $I$
in time $O(\eta(I)^{c+1}) r^{\mu(I)}$.
\end{lemma}

A reduction rule that creates at most one subinstance is a \emph{simplification} rule, other reduction rules are \emph{branching} rules.

\section{\mcsp}
\label{sec:csp}

Using the notation from \cite{ScottS07}, an instance $(G, \mathcal{S})$ of \mcsp (also called \mrcsp) is given by a \emph{constraint graph} $G = (V, E)$ and a set $\mathcal{S}$ of \emph{score functions}.
Writing $[r] = \{1, . . . , r\}$ for the set of available vertex colors, we have a \emph{dyadic} score function $s_e : [r]^2 \rightarrow \mathbb{R}$ for each edge $e \in E$, a \emph{monadic} score function $s_v : [r] \rightarrow \mathbb{R}$ for each vertex $v \in V$, and a single \emph{niladic} score ``function'' $s_{\emptyset} : [r]^0 \rightarrow \mathbb{R}$ which takes no arguments and is just a constant convenient for bookkeeping.
A \emph{candidate solution} is a function $\phi : V \rightarrow [r]$ assigning colors to the vertices ($\phi$ is an \emph{assignment} or
\emph{coloring}), and its score is
\begin{align}
 s(\phi) := s_{\emptyset} + \sum_{v\in V} s_v(\phi(v)) + \sum_{uv \in E} s_{uv} (\phi(u), \phi(v)).
  \label{score}
\end{align}
An \emph{optimal solution} $\phi$ is one which maximizes $s(\phi)$.

Let us recall the reductions from \cite{ScottS07}. \Red{0, 1, and 2} are
simplification rules (acting, respectively, on vertices of degree 0, 1, and 2),
creating one subinstance, and \Red{3} is a branching rule, creating $r$ subinstances. An optimal solution for $(G,\mathcal{S})$ can be found in polynomial time from optimal solutions of the subinstances.
\begin{description}
\item[\Red{0}] If $d(y)=0$, then set $s_{\emptyset} = s_{\emptyset} + \max_{C\in [r]} s_y(C)$ and delete $y$ from $G$.
\item[\Red{1}] If $N(y)=\{x\}$, then replace the instance with $(G',\mathcal{S}')$ where $G'=(V',E') = G-y$ and $\mathcal{S}'$ is the restriction of $\mathcal{S}$ to $V'$ and $E'$ except that for all colors $C\in [r]$ we set
\begin{align*}
s'_x(C) = s_x(C) + \max_{D\in [r]} \{ s_{xy}(C,D) + s_y(D) \}.
\end{align*}
\item[\Red{2}] If $N(y)=\{x,z\}$, then replace the instance with $(G',\mathcal{S}')$ where $G'=(V',E') = (V-y, (E\setminus \{x y, y z\})\cup \{x z\})$ and $\mathcal{S}'$ is the restriction of $\mathcal{S}$ to $V'$ and $E'$, except that for $C,D\in [r]$ we set
\begin{align*}
s'_{xz}(C,D) = s_{xz}(C,D) + \max_{F\in[r]} \{ s_{xy}(C,F) + s_{yz}(F,D) + s_y(F) \}
\end{align*}
if there was already an edge $x z$, discarding the first term $s_{x z}(C,D)$ if there was not.
\item[\Red{3}] Let $y$ be a vertex of degree at least 3. There is one subinstance $(G',s^C)$ for each color $C\in [r]$, where $G'=(V',E')=G-y$ and $s^C$ is the restriction of $s$ to $V'$ and $E'$, except that we set
\begin{align*}
(s^C)_{\emptyset} = s_{\emptyset} + s_y(C),
\end{align*}
and, for every neighbor $x$ of $y$ and every $D\in [r]$,
\begin{align*}
(s^C)_x (D) = s_x (D) + s_{x y}(D,C).
\end{align*}
\end{description}
We will now describe a new separator-based algorithm for cubic \mcsp, outperforming the algorithm given by Scott and Sorkin in \cite{ScottS07}.
Using it as a subroutine in the algorithm for general instances by Scott and Sorkin \cite{ScottS07} also gives
a faster running time for \mcsp on arbitrary graphs.

\subsection{Background}
For a cubic instance of \mcsp,
an instance whose constraint graph $G$ is 3-regular,
the fastest known polynomial-space algorithm
makes simple use of the reductions above.
The algorithm branches on a vertex $v$ of degree 3,
giving $r$ instances with a common constraint graph $G'$,
where $v$ has been deleted from the original constraint graph.
In $G'$, the three $G$-neighbors of $v$ each have degree 2.
Simplification rules are applied to rid $G'$ of degree-2 vertices,
and further vertices of degree 0, 1, or 2 that may result,
until the constraint graph becomes another cubic graph $G''$.
This
results in $r$ instances with the common constraint graph $G''$,
to which the same algorithm is applied recursively.
The running time of the algorithm is exponential in the number of branchings,
and since each branching destroys four degree-3 vertices
(the pivot vertex $v$ and its three neighbors),
the running time is bounded by $\Ostar{r^{n/4}}$;
details may be found in \cite{linear}.

For this algorithm, $r^{n/4}$ is also a lower bound,
achieved for example by an instance whose constraint graph consists
of disjoint copies of $K_4$, the complete graph on 4 vertices.
In fact, instances with disjoint constraint graphs can be
solved with far greater efficiency,
since the components can be solved separately,
a fact exploited by the fastest polynomial-space algorithms
for \mcsp \cite{ScottS07}.
However, Subsection~\ref{subsec:lb} shows that $r^{n/4}$ is also a lower bound
for algorithms exploiting connectedness, since a slightly unlucky choice
of pivot vertices leaves a particular worst-case constraint graph connected.
We conjecture that $r^{n/4}$ is a lower bound for any algorithm
using these reductions and choosing its pivot ``locally'':
that is, characterizing each vertex by the structure of the
constraint graph (or even the instance) within a fixed-radius ball around it,
and choosing a vertex with a best such character.

Here, we show how to break this $r^{n/4}$ barrier,
by selecting pivot vertices using global properties of the constraint graph.
In this section we describe an algorithm which pivots only
on vertices in a separator of $G$.
When the separator is exhausted, $G$ has been split into two components $L$
and $R$
which can be solved independently, and are treated recursively.
The efficiency gain of the algorithm comes from the component splitting:
if the time to solve an instance with $n$ vertices can be bounded by
$\Ostar{r^{cn}}$,
the time to solve an instance consisting of components $L$ and $R$
is $\Ostar{r^{c|L|}}+\Ostar{r^{c|R|}}$,
which (for $L$ and $R$ of comparable sizes) is
hugely less than the time bound $\Ostar{r^{c(|L|+|R|)}}$
for a single component of the same total order.
This efficiency gain comes at no cost:
until the separator is exhausted,
branching on vertices in the separator is just as efficient as branching
on any other vertex.

\subsection{Algorithm and Analysis}
We interleave the algorithm's description with its running time analysis.
To analyze the algorithm, we use the \mc method.
As in \cite{hybrid}, we use penalty terms in the measure to treat tricky cases
that otherwise require arguments outside the \mc framework.
(Those arguments are typically simple but mesh poorly with the
\mc framework, making correctness difficult to check.
Our penalty approach is modeled on Wahlstr{\"o}m's \cite{Wahlstrom04}.)
We also take from \cite{ScottS07} and \cite{hybrid}
the treatment of vertices of degrees 1 and 2 within the \mc
framework.%
Reductions on such vertices are \textit{de facto} never an algorithm's critical cases,
but can lead to a tangle of special cases unless treated uniformly.

Recall (see Lemma \ref{lem:measureanalysis}) that the \mc analysis
applies to an algorithm which transforms an
instance $I$ to one or more instances $I_1,\ldots,I_k$ in polynomial,
solves those instances recursively,
and obtains a solution to $I$ in polynomial time from
the solutions of $I_1, \ldots, I_k$.
By definition, the measure $\mu(I)$ of an instance $I$ should satisfy
that for any instance,
\begin{align}
\mu(I) & \geq 0 .  \label{mupos}
\end{align}
To satisfy constraint \eqref{eq:masize} of Lemma \ref{lem:measureanalysis} we should have that
for any transformation of $I$ into $I_1,\ldots,I_k$,
\begin{align}
r^{\mu(I_1)} + \cdots + r^{\mu(I_k)} & \leq r^{\mu(I)} .     \label{trans}
\end{align}
Given these hypotheses,
the algorithm solves any instance $I$ in time
$\Ostar{r^{\mu(I)}}$
if the number of recursive calls from the root to a leaf of the search tree is polynomially bounded.

Here, we present an instance of \mcsp in terms of a
separation $(L,S,R)$ of its constraint graph $G=(V,E)$.
We write $L_3$, $S_3$, and $R_3$ for the subsets of degree-3 vertices of $L$, $S$, and $R$, respectively,
and we will always assume that $\sl \leq \sr$,
if necessary swapping the roles of $L$ and $R$ to make it so.
We write $\sst$ for the number of degree-2 vertices in $S$.

We define the measure of an instance as
\begin{align}
 \mu(L,S,R) &=
   \ws \ss + \wst \sst + \wr \sr + \wb \one(\sr=\sl) \notag \\ & \mytab
     + \wc \one(\sr=\sl+1)
    + \wdd \log_{3/2} (\sr+\ss) ,
\end{align}
\noindent
where the values $\ws$, $\wst$, $\wr$, $\wb$, $\wc$, and $\wdd$
are constants to be determined
and the indicator function $\one(\text{event})$ takes the value 1 if the event is true and 0 otherwise.
In what follows, we describe the algorithm and derive which constraints inequality \eqref{trans} induces.
Notice that if we choose the values $\ws$, $\wst$, $\wr$, $\wb$, $\wc$, and $\wdd$ so that all the constraints are satisfied, then Lemma \ref{lem:measureanalysis} implies an algorithm running in time $\Ostar{r^{\mu(I)}}$.
We optimize $\mu(I)$ under the derived constraints.

For the constraint \eqref{mupos} that $\mu \geq 0$,
it suffices to constrain each of the constants to be nonnegative:
\begin{align}
 \ws, \wst, \wr, \wb, \wc, \wdd \ge 0 . \label{nonneg}
\end{align}

\noindent
Intuitively, the terms $\wb$ and $\wc$ are the only representations
of the size of $L$ in $\mu$,
and account for the greater time needed when the
left side is as large (or nearly as large) as the right.
(We thus anticipate that optimally setting the weights gives $\wb \geq \wc$,
as turns out to be true, due to constraint \eqref{degredL}.)
%
The logarithmic term 
offsets
increases in penalty terms that may result when a new separator is computed,
where the instance may go from imbalanced to balanced.


\ignore{%
One idea behind this is that if $L$ is smaller than $R$,
it suffices to bound everything in terms of $R$.
On the other hand, intuitively, a reduction like that of
Figure~\ref{figureTwoEdgesIntoLeft}
has benefit if the two sides are kept in perfect balance,
but not if the left side is significantly smaller than the right.
The purpose of $\wb$ and $\wc$ is to explicitly reward a partition's
Since $\wb$
We will take advantage of our intuition to presume that $\wb \geq \wc$,
but this could be avoided
}

Concretely,
 from \eqref{trans}, each reduction imposes a constraint on the
measure.
Whenever there is a vertex of degree 0, 1, or 2
a corresponding reduction is applied,
so it can be presumed that any other reduction is applied to a
cubic graph.
We treat the reductions in their order of priority:
when presenting one reduction, we assume that no previousy presented
reduction can be applied to the instance. Denote by $\mu$ the value of the measure before the reduction is applied and by $\mu'$ its value after the reduction.

\redsec{Degree 0.} \label{red0}
If the instance contains a vertex $v$ of degree $0$, then
perform \Red{0} on $v$. Removing $v$ from the instance has no effect on the measure and Condition \eqref{trans} is satisfied.

\redsec{Half-edge deletion.}
If there is a vertex of degree 1, apply \Red{1} on it, and if there is a vertex in $L\cup R$ of degree 2, apply \Red{2} on it.
We may view the deletion of an edge $uv$ as a deletion of two half-edges, one incident on $u$ and the other incident on $v$. Deleting the half-edge incident on $v$ also decreases the degree of $v$, but not the degree of $u$.
Similarly, when a reduction creates an edge between two vertices that already had an edge, our analysis considers that a temporary parallel edge emerges, which then collapses into a single edge.

We will require that a half-edge deletion does not increase the measure.
This will also ensure that the collapse of parallel edges, \Red{1},
and \Red{2} for vertices in $L\cup R$ does not increase the measure, which also validates
Condition \eqref{trans} for \Red{1},
and \Red{2} for vertices in $L\cup R$.
The only way that a half-edge deletion affects the measure is that the degree of a vertex $v$ decreases.

First, assume $d(v)=1$. The degree of $v$ is reduced to $0$. Since neither degree-0 nor degree-1 vertices affect the measure, $\mu'-\mu=0$, satisfying Condition \eqref{trans}.

Next, assume $d(v)=2$.
If $v\in L$, then $\mu'-\mu=0$.
If $v\in S$, then $\mu'-\mu \le -\wst$, which satisfies Condition \eqref{trans} since $\wst\ge 0$ by \eqref{nonneg}.
If $v\in R$, then $\mu'-\mu = 0$, which also satisfies Condition \eqref{trans}.

Finally, assume $d(v)=3$. We consider three cases.
\begin{itemize}
 \item $v\in L$. Since the imbalance increases by one vertex, $\mu'-\mu \le \max(0,-\wc,-\wb+\wc)$. Since $\wc\ge 0$ by \eqref{nonneg} it suffices to constrain
  \begin{align}
   -\wb+\wc \le 0 . \label{degredL}
  \end{align}
 \item $v\in S$. For this case it is sufficient that
  \begin{align}
   -\ws+\wst \le 0. \label{degredS}
  \end{align}
 \item $v\in R$. The case where $\sr=\sl$ is covered by
  \eqref{degredL} since we can swap $L$ and $R$. Otherwise, $\sr \ge \sl+1$, and then $\mu'-\mu \le -\wr+\max(0,\wc,\wb-\wc)$. Since $\wr\ge 0$ by \eqref{nonneg}, it suffices to constrain
  \begin{align}
   -\wr+\wc &\le 0 \text{ and} \label{degredR1}\\
   -\wr+\wb-\wc &\le 0 . \label{degredR2}
  \end{align}
\end{itemize}

\redsec{Separation.}
This reduction is the only one special to separation,
and its constraint looks quite different from those in previous works.
The reduction applies when $S=\emptyset$, which arises in two cases.
One is at the beginning of the algorithm, when the instance has not been
separated, and may be represented by the trivial separation
$(\emptyset,\emptyset,V)$.
The second is when reductions on separated instances have exhausted the
separator, so that $S$ is empty but $L$ and $R$ are nonempty,
and the instance is solved by
solving the instances on $L$ and $R$ independently,
via a new separation $(L',S',R')$ for $R$
and another such separation $(L'',S'',R'')$ for $L$.
The reduction is applied to a graph $G=(V,E)$ that is cubic and
can be assumed to be of at least some constant order, $\sv \geq k$,
since a smaller instance can be solved in constant time.
By Lemma \ref{lem:sep-cubic} we know that, for any constant $\e>0$,
there is a size $k=k(\e)$ such that any cubic graph $G$ of order at least $k$
has a separation $(L,S,R)$ with
$|S| \leq (\frac16+\e) \sv$ , $|L|, |R| \leq \frac5{12} \sv$.
To satisfy \eqref{trans}, we will make worst-case assumptions about balance,
namely that the instance goes from being imbalanced ($\sr\ge \sl+2$) to being balanced ($\sr[']=\sl[']$ and $\sr['']=\sl['']$).
(Note that $\wb \ge \wc \ge 0$ by \eqref{nonneg} and \eqref{degredL}.)
It thus suffices to constrain that
\begin{align*}
r^{ \ws \ss[']+ \wr \sr['] + \wb + \wdd \log_{3/2}(\sr[']+\ss['])}
+ &
r^{ \ws \ss['']+ \wr \sr[''] + \wb + \wdd \log_{3/2}(\sr['']+\ss[''])}
\\ &
\leq
r^{  \wr \sr + \wdd \log_{3/2}(\sr)} .
\end{align*}
From the separator properties, this in turn is implied by
\begin{align*}
2 \cdot r^{ \ws (1/6+\e) \sr + \wr(5/12\cdot \sr) + \wb
   + \wdd\log_{3/2}(8/12 \cdot \sr)}
& \leq
r^{  \wr \sr
   + \wdd\log_{3/2}(\sr)} ,
\end{align*}
where we have estimated
$\sl['], \sr['] \leq \tfrac5{12} \sr$
and $\ss['] \leq (\tfrac16+\e) \sr \le \tfrac3{12} \sr$ in the $\log$ term on the left hand side.
Since $r\ge 2$, it suffices to constrain that
\begin{align*}
1+ \ws (\tfrac16+\e) \sr + \wr(\tfrac5{12} \sr) + \wb
   + \wdd\, \log_{3/2}(\tfrac8{12} \sr)
& \leq
\wr \sr + \wdd\, \log_{3/2}(\sr).
\end{align*}
Since we took $\tfrac32 = \tfrac{12}8$ to be the logarithm's base, we set $\wdd=\wb+1$ so that
the term $\wdd\log_{3/2}(\tfrac8{12}\sr)$
is equal to $-(\wb+1) + (\wb+1) \log_{3/2}(\sr)$,
and it suffices to have
\begin{align}
 (\tfrac16+\e) \ws + \tfrac5{12} \wr
& \leq
 \wr .          \label{r1}
\end{align}

\redsec{Degree 2 in $S$.} \label{redII}
If the instance has a vertex $s\in S$ of degree 2, then perform \Red{2} on $s$. Let $u_1,u_2$ denote the neighbors of $s$.
The vertex $s$ is removed and the edge $u_1 u_2$ is added if it was not present already.
If $L$ or $R$ contain no neighbor of $s$, Condition \eqref{trans} is implied by the constraints of the half-edge deletions.
If $u_1\in L$ and $u_2\in R$ (or the symmetric case), then $S$ is not a separator any more.
The algorithm removes $u_2$ from $R$ and adds it to $S$.
If $d(u_2)=2$, we have that $\mu'-\mu \le 0$. Otherwise, $d(u_2)=3$.
If initially we had $\sl=\sr$, then $L$ and $R$ will be swapped after the reduction, and we constrain that
\begin{align}
 -\wst+\ws-\wb+\wc \le 0. \label{2S}
\end{align}
Otherwise, $\mu'-\mu \le -\wst+\ws-\wr+\max(0,\wc,\wb-\wc)$. Since $\wc\ge 0$ by \eqref{nonneg} it suffices to constrain

\begin{align}
 -\wst+\ws-\wr+\wc &\le 0 , \text{ and} \label{2S1}\\
 -\wst+\ws-\wr+\wb-\wc &\le 0 . \label{2S0}
\end{align}

\medskip
In the remaining cases, every vertex has degree 3.

\redsec{No neighbor in $L$.} \label{noL}
If the separation $(L,S,R)$ has a vertex $v \in S$ with no neighbor in $L$,
``drag'' $v$ into $R$, i.e., transform the instance by changing the
separation to $(L',S',R') := (L, S\setminus \set{v}, R \cup \set v)$.
It is easily checked that this is a valid separation,
with no edge incident on both $L$ and $R$, 
and with $\sl['] \leq \sr[']$ implied by $\sl \leq \sr$.
Indeed the new instance is no more balanced than the old,
so that the difference between the new and old measures is
$\mu'-\mu \leq -\ws + \wr$,
and to satisfy condition \eqref{trans} it suffices that
\begin{align}
 -\ws+\wr & \leq 0,          \label{r2}
\end{align}
since
by \eqref{nonneg} and \eqref{degredL}
an increase in imbalance does not increase the measure.


\redsec{No neighbor in $R$.}
This case is similar to the previous case, but a vertex $v \in S$ with no neighbor in $R$
is dragged into $L$.
\begin{itemize}
\item
If initially we had $\sr=\sl$, we reverse the roles of $L$ and $R$
and revert to the previous case. 
\item
If $\sr\ge \sl+1$ then the transformation increases $\sl$ by 1, decreasing the imbalance by one vertex. Therefore,
$\mu'-\mu \le -\ws+\max(0,\wc,\wb-\wc)$ and
we constrain that
\begin{align}
 -\ws+\wc &\le 0, \text{ and} \label{noR2}\\
 -\ws+\wb-\wc &\le 0 . \label{noR1}
\end{align}
\end{itemize}

With the above cases covered, we may assume that the pivot vertex $s \in S$ has degree 3 and
at least one neighbor in each of $L$ and $R$.

	\begin{figure}[tbp]
		\centering
		\begin{subfigure}[b]{0.33\textwidth}
			\centering
			\begin{tikzpicture}[scale=1]
			\draw (0,0) ellipse (0.5cm and 1cm);
			\node at (0,1.2) {$S$};
			\draw (-1.2,0) ellipse (0.6cm and 1cm);
			\node at (-1.2,1.2) {$L$};
			\draw (1.2,0) ellipse (0.6cm and 1cm);
			\node at (1.2,1.2) {$R$};
			
			\draw (0,0) node[vertex,label=above:$s$] (s) {};
			\draw (0,-0.6) node[vertex] (s2) {};
			\draw (-1,0) node[vertex] (l) {};
			\draw (1,0) node[vertex] (r) {};
			
			\draw (l)--(s)--(r) (s)--(s2);
			\end{tikzpicture}
			\caption{\label{fig:lsr}One neighbor in $L$, $S$, and $R$}
		\end{subfigure}
		\quad
		\begin{subfigure}[b]{0.3\textwidth}
			\centering
			\begin{tikzpicture}[scale=1]
			\draw (0,0) ellipse (0.5cm and 1cm);
			\node at (0,1.2) {$S$};
			\draw (-1.2,0) ellipse (0.6cm and 1cm);
			\node at (-1.2,1.2) {$L$};
			\draw (1.2,0) ellipse (0.6cm and 1cm);
			\node at (1.2,1.2) {$R$};
			
			\draw (0,0) node[vertex,label=above:$s$] (s) {};
			\draw (-1,0.3) node[vertex] (l) {};
			\draw (-1,-0.3) node[vertex] (l2) {};
			\draw (1,0) node[vertex] (r) {};
			
			\draw (l)--(s)--(r) (s)--(l2);
			\end{tikzpicture}
			\caption{\label{fig:llr}Two neighbors in $L$}
		\end{subfigure}
		\quad
		\begin{subfigure}[b]{0.3\textwidth}
			\centering
			\begin{tikzpicture}[scale=1]
			\draw (0,0) ellipse (0.5cm and 1cm);
			\node at (0,1.2) {$S$};
			\draw (-1.2,0) ellipse (0.6cm and 1cm);
			\node at (-1.2,1.2) {$L$};
			\draw (1.2,0) ellipse (0.6cm and 1cm);
			\node at (1.2,1.2) {$R$};
			
			\draw (0,0) node[vertex,label=above:$s$] (s) {};
			\draw (-1,0) node[vertex] (l) {};
			\draw (1,0.3) node[vertex] (r) {};
			\draw (1,-0.3) node[vertex] (r2) {};
			
			\draw (l)--(s)--(r) (s)--(r2);
			\end{tikzpicture}
			\caption{\label{fig:lrr}Two neighbors in $R$}
		\end{subfigure}
		\caption{\label{fig:2csp-cases} Configurations for branching on a separator vertex.}
	\end{figure}
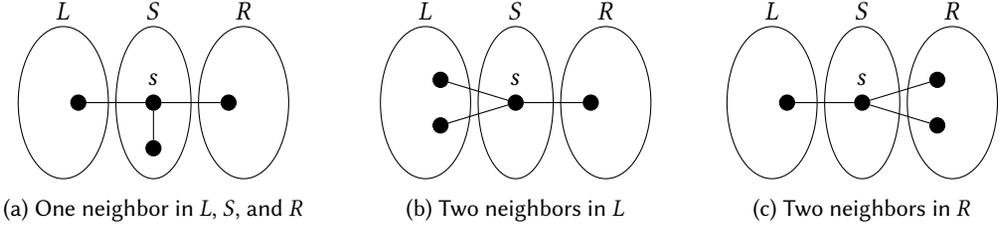

\redsec{One neighbor in each of $L$, $S$, and $R$.}
If there is a vertex $s \in S$ with one neighbor in each of $L$, $S$, and $R$ (Figure \ref{fig:lsr}),
perform \Red{3} on $s$, deleting it from the constraint graph
and thereby reducing the degree of each neighbor to $2$.
Since both $L$ and $R$ lose a degree-$3$ vertex, there is no
change in balance and the constraint is
\begin{align}
 1 -2\ws + \wst -\wr \leq 0 .   \label{r5}
\end{align}
The form and the initial 1 come from the reduction's generating $r$ instances
with common measure $\mu'$, so the constraint is
$r \cdot r^{\mu'} \leq r^{\mu}$, or equivalently $1+\mu' -\mu \leq 0$.
The value of $\mu'-\mu$ comes from
$S$ losing two degree-3 vertices but gaining a degree-2 vertex, and $R$ losing a degree-3 vertex.

\redsec{Two neighbors in $L$.}
If $s\in S$ has two neighbors in $L$ and one neighbor in $R$  (Figure \ref{fig:llr}), applying \Red{3} removes $s$, reduces the degree of
a degree-3 vertex in $R$, and increases the imbalance by
one vertex.
The algorithm performs \Red{3} if $\sr \le \sl+1$, where $\mu'-\mu \le -\ws-\wr+\max(-\wb+\wc, -\wc)$. Thus, we constrain%

  \begin{align}
   1-\ws-\wr-\wb+\wc &\le 0 \text{ and} \label{red2L0}\\
   1-\ws-\wr-\wc &\le 0 . \label{red2L1}
  \end{align}

If, instead, $\sr \ge \sl+2$, then the algorithm drags $s$ into $L$ and its neighbor $r\in R$ into $S$, replacing $(L,S,R)$ by $(L\cup\set{s}, (S\setminus \set{s})\cup \set{r}, R\setminus \set{r})$. We need to ensure that $-\wr+\max(\wb,\wc) \le 0$, which, since $\wc \le \wb$ by \eqref{degredL}, is satisfied if we constrain that
  \begin{align}
   -\wr+\wb \le 0 . \label{red2L2}
  \end{align}

\redsec{Two neighbors in $R$.}
If $s\in S$ has two neighbors in $R$ and one neighbor in $L$ (Figure \ref{fig:lrr}), the algorithm performs \Red{3}, which removes $s$, reduces the degree of
two degree-3 vertices in $R$, and decreases the imbalance by one.
For the analysis of the case where $\sr=\sl$, we refer to  \eqref{red2L0} since $L$ and $R$ are swapped after the reduction. For the other cases, we constrain

\begin{align}
 1-\ws-2\wr-\wc+\wb &\le 0 \text{ and} \label{red2R1}\\
 1-\ws-2\wr+\wc &\le 0 . \label{red2R2}
\end{align}

\medskip
\noindent
This concludes the description of the algorithm and describes all the constraints on the measure.
To minimize the running time proven by the analysis
(see after \eqref{mupos}, \eqref{trans}),
we minimize the initial measure, which is $\mu(\emptyset,\emptyset,V) =
\wr |V| + \wdd \log_{3/2} (|V|)$.
Since the logarithmic term affects the running time only by a polynomial,
we therefore minimize $\wr$, subject to our constraints \eqref{nonneg}--\eqref{red2R2}, which are all linear.
We obtain the following optimal, feasible weights using linear programming:
\begin{align*}
 \wr  &= 0.2 + o(1) & \ws  &= 0.7 & \wb &= 0.2\\
      &      & \wst &= 0.6 & \wc &= 0.1
\end{align*}

All constraints are satisfied and $\mu = (1/5+o(1))n$.
The tight constraints in our analysis are \eqref{r1}--\eqref{2S0} and \eqref{r5}--\eqref{red2R2}.

It only remains to verify that the depth of the search trees of the
algorithm is upper bounded by a polynomial. Since not every reduction
removes a vertex (some only modify the separation $(L,S,R)$),
it is crucial to guarantee some kind of progress for each reduction.
We will argue that each reduction decreases another polynomially-bounded measure $\eta(L,S,R,E) := 3 \ss + 2 \sr + \sl + 2 |E|$,
by at least one, and
the depth
of the search trees is therefore polynomial.
For those reductions that remove one or more vertices, it is easily seen that $\eta(L,S,R,E)$ decreases by at least one.
In the Separation case, the instance $(L,\emptyset,R)$ leads to two instances $(\emptyset,\emptyset,L)$ and $(\emptyset,\emptyset,R)$ in a first step, and for each of these two instances, the measure decreases by at least one, since $|R_3|\ge |L_3|$. In a second step, an instance $(\emptyset,\emptyset,R)$ is replaced by $(L',S',R')$ where the measure decreases by at least one since $|L'|>|S'|$.
For those reductions where only the partition $(L,S,R)$ is modified, it suffices to note that either a vertex is moved from $S$ to $L\cup R$ (in the cases No neighbor in $L$ and No neighbor in $R$) or that a degree-3 vertex is moved from $R$ to $S$ and a degree-3 vertex is moved from $S$ to $L$ (Two neighbors in $L$ with $|R_3|\ge |L_3|+2$).

\subsection{MAX 2-CSP result, consequences, and extensions} 

The previous section established that our algorithm solves cubic instances
of \mcsp in time $r^{n/5+o(n)}$.
The algorithm uses only polynomial space,
since it is recursive with polynomial recursion depth
and uses polynomial space in each recursive call.

The polynomial-space \mcsp algorithm described in \cite{ScottS07},
for solving an arbitrary instance with $n$ vertices and $m$
edges in time $\Ostar{r^{19m/100}}$, 
relied upon the ability to solve cubic cases in time $\Ostar{r^{m/6}}$.
A theorem in the same work,
\cite[Theorem 22]{ScottS07},
quantifies how speedups in solving the cubic case translate into
speedups for general instances and for instances with constraint graphs of maximum degree~4,
and shows that if the cubic solver is polynomial-space
then so is the general algorithm.
For instance, for maximum degree 4 graphs, Theorem 22 of \cite{ScottS07} states that an algorithm for cubic graphs with running time $O^*(r^{\alpha m})$ leads to an algorithm for maximum degree 4 graphs with running time $O^*(r^{\beta_4(\alpha) m})$, where
\begin{align*}
	\beta_4(\alpha) = \begin{cases}
		1/8 + (3/8) \alpha & 1/9 \le \alpha \le 1/5\\
		1/6 & 0 \le \alpha \le 1/9 .
	\end{cases}
\end{align*}
Thus, our cubic algorithm's running time of $O^*(r^{(2/15+\epsilon) m})$, for any $\epsilon>0$, leads to an algorithm with running time $O^*(r^{(7/40 +\epsilon') m})$ for arbitrarily small $\epsilon'>0$ for maximum degree 4 graphs.

\begin{theorem} \label{CSPthm}
On input of a \mcsp instance on a constraint graph $G$
with $n$ vertices and $m$ edges,
the described algorithm solves $G$ in time
 $r^{n/5+o(n)} = r^{2m/15+o(m)}$ if $G$ is cubic,
 time $r^{7m/40+o(m)}$ if $G$ has maximum degree~4, and
 time $r^{9m/50+o(m)}$ in general,
 {in all cases }using {only }polynomial space.
\end{theorem}
\noindent
{%
This improves respectively on the previous best polynomial-space running times of
$\Ostar{r^{m/6}}$, $\Ostar{r^{3m/16}}$, and
$\Ostar{r^{19m/100}}$, all from \cite{ScottS07}.
It also improves on the fastest known polynomial-space running times for \textsc{Max Cut} on cubic, maximum degree~4, and general graphs.
The \textsc{Max Cut} problem, a special case of \mtcsp,
is to find a bipartition of the vertices of a given graph maximizing the number of edges crossing the bipartition.
Its previously fastest polynomial-space algorithms were the one by Scott and Sorkin \cite{ScottS07} for cubic, maximum degree~4, and general graphs.

Theorem \ref{CSPthm} extends instantly to
generalizations of \mcsp introduced by Scott and Sorkin in \cite{PCSP}:
Ring CSP (RCSP), where the scores take values in an arbitrary ring,
and Polynomial CSP (PCSP), where the scores are multivariate formal polynomials.
In analogy with the definition of a \mcsp instance in and around \eqref{score},
an instance $I=(G,\mathcal{S})$ of RCSP over a ring $R$ and domain $[r]$ is
composed of a graph $G=(V,E)$ and a set $\mathcal{S}$ of score functions:
a \emph{dyadic} score function $s_e : [r]^2 \rightarrow R$ for each edge $e \in E$, a \emph{monadic} score function $s_v : [r] \rightarrow R$ for each vertex $v \in V$, and a \emph{niladic} score function $s_{\emptyset} : [r]^0 \rightarrow R$.
The score of an assignment $\phi : V \rightarrow [r]$ is the ring element
\begin{align*}
s(\phi) := s_{\emptyset} \cdot \prod_{v\in V} s_v(\phi(v)) \cdot \prod_{uv \in E} s_{uv} (\phi(u), \phi(v)),
\end{align*}
and the RCSP problem is to compute the \emph{partition function}
\begin{align*}
Z_I = \sum_{\phi : V \rightarrow [r]} s(\phi).
\end{align*}
Comparing with \mcsp, in RCSP the scores are ring-valued (rather than real-valued),
the scores are multiplied (rather than added),
and the solution is the sum of all assignment scores (rather than the maximum).
(In fact, for our purposes the ``ring'' in RCSP can be relaxed to a semiring,
like a ring but lacking negation,
and \mcsp may be viewed as the case of the semi-ring $R$ over the reals,
where the product operation for $R$ is the real sum,
the addition operation for $R$ is the maximum of the two reals,
and the zero for $R$ is $-\infty$.)

Scott and Sorkin \cite{PCSP} showed that any branching algorithm for \mcsp based solely on
\Reds{0, 1, 2, and 3} can be extended to RCSP. 
They replace \Reds{0, 1, 2, and 3} by reductions appropriate for the more general
setting and prove that each one can be executed using $O(r^3)$ ring operations,
while changing the graph in the same way as the original reductions.
The number of ring operations is thus of the same order as the number of instances in the \mcsp branching algorithm.
Where Theorem~4 of \cite{PCSP} applied this observation to ``Algorithm~B'' of \cite{ScottS07},
applying the observation to our algorithm yields the following theorem.

\begin{theorem}
Let $R$ be a ring.
On input of an RCSP instance $I$ over $R$
with domain $[r]$, $n$ variables, $m$ constraints, and constraint graph $G$,
the ring extension of the algorithm described here
calculates the partition function $Z_I$
of $I$ in polynomial space
and with $r^{n/5+o(n)}$ ring operations if $G$ is cubic,
$r^{7m/40+o(m)}$ ring operations if $G$ has maximum degree~4, and
$r^{9m/50+o(m)}$ ring operations in general.
\end{theorem}

\cite{PCSP} defines PCSP as
the special case of RCSP where the ring $R$ is a polynomial ring over the reals.
The polynomials may be multivariate, and indeed may have negative and fractional powers.
The obvious use of this is for generating functions.
For example, to count the number of cuts $a_{ij}$ of a graph with $i$ vertices in side ``0'' 
and $j$ cut edges,
take the score of each vertex $v$ to be $x^0$ if $\phi(v)=0$ and $x^1$ if $\phi(v)=1$,
and the score of each edge $(u,v)$ to be $y^0$ if $\phi(u)=\phi(v)$ and $y^1$ if $\phi(u) \neq \phi(v)$;
then the partition function $Z=\sum a_{ij} x^i y^j$ gives all the desired counts $a_{ij}$.
For further details and other applications see \cite{PCSP}.
As a special case of RCSP, PCSP can of course be solved by the same algorithm
in the same number of ring operations,
but here our interest is in the total running time,
so the times to multiply and add the polynomials must be taken into account.
The details of this are beyond the scope of this article, but
\cite{PCSP} defines ``polynomially bounded'' PCSP instances,
where the ring operations can be performed efficiently,
and a more general class of ``prunable'' PCSP instances, where
the leading term of the partition function can be computed efficiently
(even though the full polynomial may not be efficiently computable,
and may for example be of size exponential in the input size).
In analogy with \cite[Theorem~12]{PCSP},
this yields the following theorem.

\begin{theorem}
The PCSP extension of our algorithm solves any polynomially bounded PCSP instance
(or finds the pruned partition function of any prunable PCSP instance)
with domain $[r]$, $n$ variables, $m$ constraints, and constraint graph $G$,
in polynomial space and in
time $r^{n/5+o(n)}$ if $G$ is cubic,
time $r^{7m/40+o(m)}$ if $G$ has maximum degree~4, and
time $r^{9m/50+o(m)}$ in general.
\end{theorem}

Consequences include the ability to efficiently solve graph bisection,
count Max $r$-SAT solutions,
find judicious partitions,
and compute the Ising partition function;
the formulation of these and other problems as PCSPs is given in \cite{PCSP}.
}%
{

Edwards \cite{Edwards16} derived, using a result from \cite{EdwardsM15}, a general way to design
fast algorithms in terms of average degree based on algorithms whose running
times depend on the number of edges of the graph.
In particular, the running time for general graphs from Theorem~\ref{CSPthm} (which is matched by Edwards \cite{Edwards16})
leads to a polynomial-space algorithm solving
\mcsp instances with average degree $d\ge 2$ in time $\Ostar{r^{n\cdot \left( 1-\frac{3.4}{d+1}+O(1/d^3) \right)}}$ and polynomial space.

\begin{theorem}[\cite{Edwards16}]
	\mcsp instances with average degree $d\ge 2$ can be solved in time $\Ostar{r^{n\cdot \left( 1-\frac{3.4}{d+1}+O(1/d^3) \right)}}$ and polynomial space.
\end{theorem}

Finally, we{ would like to} highlight that any improvement on Lemma \ref{lem:sep-cubic} will automatically improve our running times.
}

{
\subsection{Lower Bounds for the Scott-Sorkin Algorithm}
\label{subsec:lb}

{%
Scott and Sorkin \cite{ScottS07} analyzed the running time of their \mcsp algorithm with respect to $m$, the number of edges of the input instance, and showed the following upper bounds:
\begin{itemize}
 \item $\Ostar{r^{m/6}}$ for subcubic instance, 
 \item $\Ostar{r^{3m/16}}$ for instances with maximum degree 4, and
 \item $\Ostar{r^{19m/100}}$ for instances with no degree restrictions.
\end{itemize}
In this subsection we prove matching lower bounds.
Define $f(n) = \Thstar{g(n)}$ if $f(n) = \Ostar{g(n)}$ and $f(n) = \Omega(g(n))$.
{%
By Lemmas \ref{lem:lbthree}--\ref{lem:lbfive}, proved hereafter, and the analysis from \cite{ScottS07},
we obtain the following theorem.
}

\begin{theorem}
 The worst-case running time of the Scott-Sorkin algorithm for \mcsp is
 \begin{itemize}
  \item $\Thstar{r^{m/6}}$ for subcubic instances, 
  \item $\Thstar{r^{3m/16}}$ for instances with maximum degree 4, and
  \item $\Thstar{r^{19m/100}}$ for instances with no degree restrictions,
 \end{itemize}
 where $m$ is the number of edges in the input instance,
 even when the input instance is connected.
\end{theorem}

{
The lower bounds are established by Lemmas \ref{lem:lbthree}--\ref{lem:lbfive}, starting with instances with maximum degree 3.
These lemmas define infinite families of instances with maximum degrees $3$, $4$, and $5$, and prove that the algorithm may perform $\Omega(r^{m/6})$, $\Omega(r^{3m/16})$, and $\Omega(r^{19m/100})$ steps, respectively.
From the dual solution in the LP analysis in \cite{ScottS07},
it follows that the $r^{19m/100}$ time bound for general instances rests on
the algorithm branching on vertices of degrees $5$, $4$, and $3$ in the
proportion $8$ to $6$ to $5$,
and that no additional \Red{0, 1, and 2} occur beyond the ones directly needed by \Red{3} (an application of \Red{3} on a vertex $v$ necessarily leads to an application of \Red{0, 1, or 2} on each neighbor of $v$).
The LP analysis does not indicate how to go about constructing instances for
which this occurs, nor prove their existence.
The running time lower bound of
Lemma~\ref{lem:lbfive} is proved by explicitly constructing such instances.

\begin{figure}[tbp]
	\begin{tikzpicture}[scale=1.1]
	\draw (0,0) node[vertex,label=above:$a_{n-3}$] (al) {};
	\draw (0.5,0.5) node[vertex,label=above:$a_{n-2}$] (bl) {};
	\draw (1,0) node[vertex,label=above:$a_{n-1}$] (cl) {};
	\draw (0.5,-0.5) node[vertex,label=below:$a_n$] (dl) {};
	\draw (2,0) node[vertex,label=above:$a_1$] (a1) {};
	\draw (2.5,0.5) node[vertex,label=above:$a_2$] (b1) {};
	\draw (3,0) node[vertex,label=above:$a_3$] (c1) {};
	\draw (2.5,-0.5) node[vertex,label=below:$a_4$] (d1) {};
	\draw (-2,0) node[vertex,label=above:$a_{n-7}$] (asl) {};
	\draw (-1.5,0.5) node[vertex,label=above:$a_{n-6}$] (bsl) {};
	\draw (-1,0) node[vertex,label=above:$a_{n-5}$] (csl) {};
	\draw (-1.5,-0.5) node[vertex,label=below:$a_{n-4}$] (dsl) {};
	
	\node at (-4,0) {$G_3(n)$};
	
	\draw (al)--(bl)--(cl)--(dl)--(al) (bl)--(dl) (cl)--(a1);
	\draw (a1)--(b1)--(c1)--(d1)--(a1) (b1)--(d1);
	\draw 
	(c1)--(3.5,0) .. controls +(0.5,0) and +(0.5,0) .. (3.5,1);
	\draw 
	(-2.5,1) .. controls +(-0.5,0) and +(-0.5,0) .. (-2.5,0)--(asl);
	\draw (asl)--(bsl)--(csl)--(dsl)--(asl) (bsl)--(dsl) (csl)--(al);
	
	\begin{scope}[yshift=-3cm]
	\draw (0,0) node[vertex,label=above:$a_{n-3}$] (al) {};
	\draw (0.5,0.5) node[vertex,label=above:$a_{n-2}$] (bl) {};
	\draw (1,0) node[vertex,label=above:$a_{n-1}$] (cl) {};
	\draw (2,0) node[vertex,label=above:$a_1$] (a1) {};
	\draw (2.5,0.5) node[vertex,label=above:$a_2$] (b1) {};
	\draw (3,0) node[vertex,label=above:$a_3$] (c1) {};
	\draw (2.5,-0.5) node[vertex,label=below:$a_4$] (d1) {};
	\draw (-2,0) node[vertex,label=above:$a_{n-7}$] (asl) {};
	\draw (-1.5,0.5) node[vertex,label=above:$a_{n-6}$] (bsl) {};
	\draw (-1,0) node[vertex,label=above:$a_{n-5}$] (csl) {};
	\draw (-1.5,-0.5) node[vertex,label=below:$a_{n-4}$] (dsl) {};
	
	\draw (al)--(bl)--(cl) (cl)--(a1);
	\draw (a1)--(b1)--(c1)--(d1)--(a1) (b1)--(d1);
	\draw 
	(c1)--(3.5,0) .. controls +(0.5,0) and +(0.5,0) .. (3.5,1);
	\draw 
	(-2.5,1) .. controls +(-0.5,0) and +(-0.5,0) .. (-2.5,0)--(asl);
	\draw (asl)--(bsl)--(csl)--(dsl)--(asl) (bsl)--(dsl) (csl)--(al);
	\end{scope}
	\begin{scope}[yshift=-5.5cm]
	\draw (2,0) node[vertex,label=above:$a_1$] (a1) {};
	\draw (2.5,0.5) node[vertex,label=above:$a_2$] (b1) {};
	\draw (3,0) node[vertex,label=above:$a_3$] (c1) {};
	\draw (2.5,-0.5) node[vertex,label=below:$a_4$] (d1) {};
	\draw (-2,0) node[vertex,label=above:$a_{n-7}$] (asl) {};
	\draw (-1.5,0.5) node[vertex,label=above:$a_{n-6}$] (bsl) {};
	\draw (-1,0) node[vertex,label=above:$a_{n-5}$] (csl) {};
	\draw (-1.5,-0.5) node[vertex,label=below:$a_{n-4}$] (dsl) {};
	
	\draw (csl)--(a1);
	\draw (a1)--(b1)--(c1)--(d1)--(a1) (b1)--(d1);
	\draw 
	(c1)--(3.5,0) .. controls +(0.5,0) and +(0.5,0) .. (3.5,1);
	\draw 
	(-2.5,1) .. controls +(-0.5,0) and +(-0.5,0) .. (-2.5,0)--(asl);
	\draw (asl)--(bsl)--(csl)--(dsl)--(asl) (bsl)--(dsl);
	\end{scope}
	
	\draw[->] (0.5,-1)--(0.5,-2) node[midway,right] {split on $a_n$};
	\draw[->] (0.5,-3.6)--(0.5,-4.9) node[midway,right] {simplify};
	\end{tikzpicture}
	\caption{\label{fig:lbcubic} The constraint graph $G_3(n)$ and how it evolves when branching on vertex $a_n$.}
\end{figure}
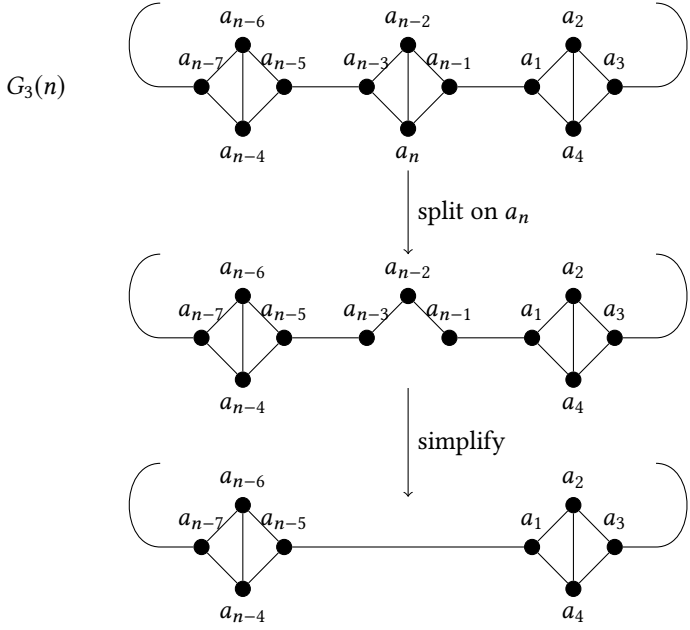

\begin{lemma}\label{lem:lbthree}
 On connected subcubic instances, the Scott-Sorkin algorithm has worst-case running time $\Omega(r^{m/6})$.
\end{lemma}
\begin{proof}
 The lemma is proven by exhibiting an infinite family of connected 3-regular instances such that the algorithm may perform $r^{n/4}$ steps when given an instance on $n$ vertices from this family.

 Let $n$ be an integer that is divisible by $4$. Consider the execution of the Scott-Sorkin algorithm on
 a \mcsp instance whose constraint graph is the graph $G_3(n)$ from Figure \ref{fig:lbcubic}. It is obtained from
 a cycle $(a_1,a_2,a_3,a_5,a_6,a_7,\ldots,a_{n-3},a_{n-2},a_{n-1})$ on $3n/4$ vertices by adding $n/4$ new vertices $a_4,a_8,\ldots,a_{n}$ and the edges $a_{4i} a_{4i-3},a_{4i} a_{4i-2},a_{4i} a_{4i-1}, 1\le i\le n/4$.

 The algorithm selects an arbitrary vertex of degree $3$ and splits on it (i.e., it performs \Red{3} on it).
 Suppose the algorithm selects $a_{n}$.

 Let us first show that the constraint graph that is obtained by performing \Red{3} on $a_{n}$ and simplifying the instance is $G_3(n-4)$.
 If $n=4$, observe that $G_3(4)$ is a complete graph on $4$ vertices.
 Branching on one vertex leaves a complete graph on $3$ vertices which vanishes by applications of \Red{2}, \Red{1} and \Red{0}.
 If $n>4$, the algorithm splits on variable $a_{n}$, which creates $r$ instances where
 $a_{n}$ is removed. Afterwards, \Red{2} is performed on the variables $a_{n-3}, a_{n-2},$ and $a_{n-1}$, and the constraint graph
 of the resulting instances is $G_3(n-4)$.

 Thus, \Red{3} is executed $n/4 = m/6$ times recursively by the algorithm, resulting in a running time of $\Omega(r^{m/6})$.
\end{proof}

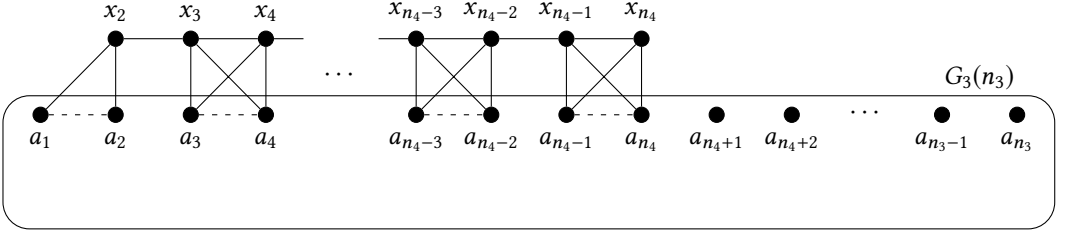
\begin{figure}[tbp]
	\begin{tikzpicture}[scale=1]
	\draw[rounded corners=10pt] (-0.5,-1.5) rectangle (13.5,0.25);
	\node at (12.5,0.5) {$G_3(n_3)$};
	
	\draw (0,0) node[vertex,label=below:$a_{1}$] (a1) {};
	\draw (1,0) node[vertex,label=below:$a_{2}$] (a2) {};
	\draw (2,0) node[vertex,label=below:$a_{3}$] (a3) {};
	\draw (3,0) node[vertex,label=below:$a_{4}$] (a4) {};
	\draw (5,0) node[vertex,label=below:$a_{n_4-3}$] (am3) {};
	\draw (6,0) node[vertex,label=below:$a_{n_4-2}$] (am2) {};
	\draw (7,0) node[vertex,label=below:$a_{n_4-1}$] (am1) {};
	\draw (8,0) node[vertex,label=below:$a_{n_4}$] (am0) {};
	\draw (9,0) node[vertex,label=below:$a_{n_4+1}$] (ap1) {};
	\draw (10,0) node[vertex,label=below:$a_{n_4+2}$] (ap2) {};
	\draw (12,0) node[vertex,label=below:$a_{n_3-1}$] (ap3) {};
	\draw (13,0) node[vertex,label=below:$a_{n_3}$] (ap4) {};
	\draw (1,1) node[vertex,label=above:$x_{2}$] (x2) {};
	\draw (2,1) node[vertex,label=above:$x_{3}$] (x3) {};
	\draw (3,1) node[vertex,label=above:$x_{4}$] (x4) {};
	\draw (5,1) node[vertex,label=above:$x_{n_4-3}$] (xm3) {};
	\draw (6,1) node[vertex,label=above:$x_{n_4-2}$] (xm2) {};
	\draw (7,1) node[vertex,label=above:$x_{n_4-1}$] (xm1) {};
	\draw (8,1) node[vertex,label=above:$x_{n_4}$] (xm0) {};
	
	\draw (a1)--(x2)--(a2) (x2)--(x3)--(a3)--(x4)--(x3)--(a4)--(x4);
	\draw 
	(x4)--(3.5,1);
	\node at (4,0.5) {$\cdots$};
	\node at (11,0) {$\cdots$};
	\draw 
	(4.5,1)--(xm3);
	\draw (xm3)--(am3)--(xm2)--(xm3)--(am2)--(xm2)--(xm1)--(am1)--(xm0)--(am0)--(xm1)--(xm0);
	\draw[mydashed] (a1)--(a2) (a3)--(a4) (am3)--(am2) (am1)--(am0);
	\end{tikzpicture}
	\caption{\label{fig:lbdegfour} The graph $G_4(n_3,n_4)$, with $n_3$ divisible by $4$ and $n_4\le n_3$ divisible by $2$,
		is obtained from $G_3(n_3)$ by removing the matching $\set{a_1a_2,\ldots,a_{n_4-1}a_{n_4}}$ (dashed edges), and adding $n_4-1$ new vertices
		$x_2,x_3,\ldots, x_{n_4}$ and the depicted edges.}
\end{figure}

\begin{lemma}\label{lem:lbfour}
 On connected instances with maximum degree 4, the Scott-Sorkin algorithm has worst-case running time $\Omega(r^{3m/16})$.
\end{lemma}
\begin{proof}
 The lemma is proven by exhibiting an infinite family of connected instances with maximum degree 4, where only a constant number of vertices have degree less than $4$, such that the algorithm may perform $\Omega(r^{3n/8})$ steps when given an instance on $n$ vertices from this family.

 The graph family will use the following construction. The graph $G_4(n_3,n_4)$, with $n_3$ divisible by $4$ and $n_4 \le n_3$ divisible by 2, is obtained from $G_3(n_3)$ by
 the following modifications. Let $M = \{a_1a_2,a_3a_4,\cdots,a_{n_4-1}a_{n_4}\}$ and observe that $M$ is a matching in $G_3(n_3)$.
 Remove the edges in $M$ from the graph.
 Add a path $(x_2,x_3,\ldots,x_{n_4})$ on $n_4-1$ new vertices to the graph, and add the edges $x_i a_i, 2\le i\le n_4$, the edge $x_2a_1$, and the edges
 $x_i a_{i-1}, x_{i-1}a_i$, for all even $i=4,6,\ldots,n_4$. See Figure \ref{fig:lbdegfour}.

 On graphs with maximum degree $4$, the Scott-Sorkin algorithm performs \Red{3} on vertices of degree $4$,
 with a preference for those vertices of degree $4$ that have neighbors of degree $3$.
 Assume that $n_4\ge 4$ and that the algorithm splits on $x_{n_4-1}$.
 We claim that this creates two instances, both with constraint graph $G_4(n_3,n_4-2)$.
 Indeed, branching on $x_{n_4-1}$ creates a constraint graph where $x_{n_4-1}$ is removed, which triggers \Red{2} on $x_{n_4}$, resulting in $G_4(n_3,n_4-2)$.
 If $n_4=2$, then \Red{2} applies to $x_2$, creating $G_4(n_3,0) = G_3(n_3)$.

 We conclude that \Red{3} is executed $n_4/2-2$ times recursively by the algorithm before reaching instances with constraint graph $G_3(n_3)$, for which
 the running time is characterized by Lemma \ref{lem:lbthree}.
 The algorithm may therefore execute
 \begin{align}
  r^{n_4/2-2} \cdot r^{n_3/4} \label{eq:lb43}
 \end{align}
 steps.
 Setting $n_3=n_4$, we obtain a running time of $\Omega(r^{(n/2-1)/2-2}\cdot r^{n/8}) = \Omega(r^{3n/8}) = \Omega(r^{3m/16})$.
\end{proof}

\begin{figure}[tbp]
	\begin{tikzpicture}[scale=0.9]
	\draw[rounded corners=10pt] (-0.9,-1.7) rectangle (2,4.7);
	\draw[rounded corners=10pt] (3.5,-1.7) rectangle (6.4,4.7);
	
	\draw (1.5,-1) node[vertex,label=left:$a_{n_4/2+n_3}$] (ul) {};
	\draw (1.5,0) node[vertex,label=left:$a_{n_4/2+n_3-1}$] (usl) {};
	\draw (1.5,2) node[vertex,label=left:$a_{n_4/2+3}$] (u3) {};
	\draw (1.5,3) node[vertex,label=left:$a_{n_4/2+2}$] (u2) {};
	\draw (1.5,4) node[vertex,label=left:$a_{n_4/2+1}$] (u1) {};
	\draw (3,-1) node[vertex,label=above:$y_{n_5}$] (xl) {};
	\draw (3,0) node[vertex,label=above:$y_{n_5-1}$] (xsl) {};
	\draw (3,2) node[vertex,label=above:$y_{3}$] (x3) {};
	\draw (3,3) node[vertex,label=above:$y_{2}$] (x2) {};
	\draw (3,4) node[vertex,label=above:$y_{1}$] (x1) {};
	\draw (4,-1.3) node[vertex,label=right:$a_{n_4/2}$] (vl1) {};
	\draw (4,-1) node[vertex,label=right:$x_{n_4/2}$] (vl2) {};
	\draw (4,-0.7) node[vertex,label=right:$a_{n_4/2-1}$] (vl3) {};
	\draw (4,-0.3) node[vertex,label=right:$x_{n_4/2-1}$] (vsl1) {};
	\draw (4,0) node[vertex,label=right:$a_{n_4/2-2}$] (vsl2) {};
	\draw (4,0.3) node[vertex,label=right:$x_{n_4/2-2}$] (vsl3) {};
	\draw (4,1.7) node[vertex,label=right:$x_{5}$] (v31) {};
	\draw (4,2) node[vertex,label=right:$a_{4}$] (v32) {};
	\draw (4,2.3) node[vertex,label=right:$x_{4}$] (v33) {};
	\draw (4,2.7) node[vertex,label=right:$a_{3}$] (v21) {};
	\draw (4,3) node[vertex,label=right:$x_{3}$] (v22) {};
	\draw (4,3.5) node[vertex,label=right:$a_{2}$] (v11) {};
	\draw (4,3.9) node[vertex,label=right:$a_{1}$] (v12) {};
	\draw (4,4.3) node[vertex,label=right:$x_{2}$] (v13) {};
	
	\draw (ul)--(xl)--(u1)--(x1)--(u2)--(x2)--(u3)--(x3) (usl)--(xsl)--(ul);
	\draw (x1)--(v11) (x1)--(v12) (x1)--(v13) (x2)--(v21) (x2)--(v22) (x2)--(v11) (x3)--(v31) (x3)--(v32) (x3)--(v33) (xsl)--(vsl1) (xsl)--(vsl2) (xsl)--(vsl3) (xl)--(vl1) (xl)--(vl2) (xl)--(vl3);
	\begin{scope}
	\clip (1.5,0) rectangle (3,0.5);
	\draw 
	(usl)--(3,1);
	\end{scope}
	\begin{scope}
	\clip (3,2) rectangle (1.5,1.5);
	\draw 
	(x3)--(1.5,1);
	\end{scope}
	\draw (usl) .. controls +(-2,1) and +(-1,-1) .. (-2,5) .. controls +(1,1) and +(1,2) .. (v12);
	\draw (u1) .. controls +(-2,-1) and +(-1,1) .. (-2,-2) .. controls +(1,-1.5) and +(1.5,-3) .. (vl3);
	
	\node at (0.7,-2.3) {$\approx G_3(n/5)$};
	\node at (4.8,-2.3) {$\approx G_4(3n/5)$};
	\end{tikzpicture}
	\caption{\label{fig:lbdegfive} The graph $G_5(n)$ is obtained from $G_4(n/2,3n/10)$ by adding an independent set $\set{y_1,\ldots,y_{n/5}}$ and the depicted edges.}
\end{figure}
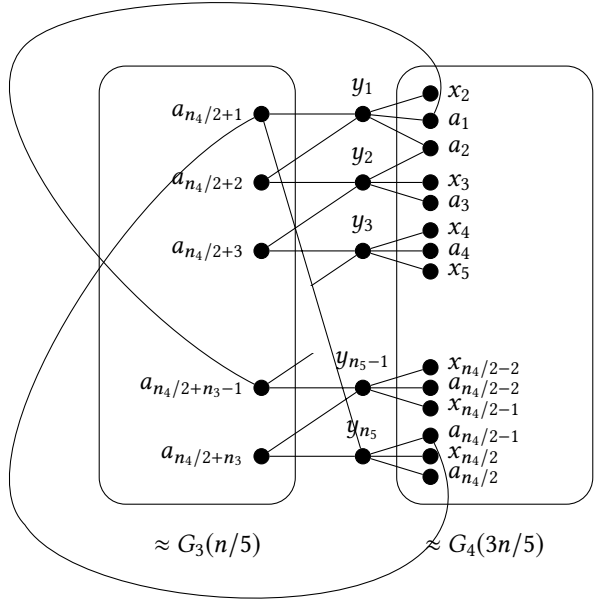

\begin{lemma}\label{lem:lbfive}
 On connected instances with maximum degree 5, the Scott-Sorkin algorithm has worst-case running time $\Omega(r^{19m/100})$.
\end{lemma}
\begin{proof}
 The lemma is proven by exhibiting an infinite family of connected instances with maximum degree 5, where only a constant number of vertices have degree less than $5$, such that the algorithm may
 perform $r^{19n/40}$ steps when given an instance on $n$ vertices from this family.

 Let $n$ be an integer that is divisible by 40.
 Let $n_3=n/5, n_4=3n/5,$ and $n_5=n/5$.
 To construct the constraint graph $G_5(n)$, we start with $G_4(n_3+n_4/2,n_4/2)$ (the construction is given in the proof of Lemma \ref{lem:lbfour}).
 Add $n_5$ new vertices $y_1,y_2, \ldots, y_{n_5}$ and edges to form a cycle $(a_{n_4/2+1}, y_1, a_{n_4/2+2}, y_2, \dots,\allowbreak a_{n_4/2+n_3}, y_{n_5}, a_{n_4/2+1})$ (observe that $a_{n_4/2+1}, \cdots, a_{n_4/2+n_3}$ all had degree 3 before adding this cycle).
 Then, for each vertex $y_i, 1\le i\le n_5$, add three incident edges, connecting $y_i$ to 3 vertices from $\set{a_1,\cdots,a_{n_4}} \cup \set{x_2,\cdots,x_{n_4}}$ in such a way that no vertex has degree greater than $5$ and $y_1$ is incident to a vertex of degree $4$.
 See Figure \ref{fig:lbdegfive}.

 On graphs with maximum degree $5$, the Scott-Sorkin algorithm performs \Red{3} on vertices of degree $5$, with a preference for those vertices of degree $5$ that have neighbors of degree $3$ or $4$.
 When there are several choices, we always assume that the algorithm splits on a $y$ vertex with minimum index.
 Observe that the algorithm splits on $y_1, y_2, \ldots, y_{n_5}$, which leaves the graph $G_4(n_3+n_4/2,n_4/2)$.
 Thus, by \eqref{eq:lb43}, the overall running time is
 \begin{align*}
  \Omega(r^{n_5} \cdot r^{(n_4/2)/2-2} \cdot r^{(n_3+n_4/2)/4}) = \Omega(r^{n/5+3n/20-2+n/20+3n/40}) = \Omega(r^{19n/40}) = \Omega(r^{19m/100})\enspace.
 \end{align*}
 This concludes the proof of the lemma.
\end{proof}

}

It should be noted that tight running time bounds are extremely rare for competitive branching algorithms.
Typically, lower bounds proved for branching algorithms are much simpler, using only one branching rule and making sure that this branching rule is applied until the instance has constant size.
In the lower bound of this section, however, we managed to make the algorithm branch according to
the mixture of branching rules leading to the upper bound ---
the mixture given by the LP's dual solution.
Our message here is twofold.
First, when designing lower bounds, it is an advantage to exploit several worst-case branchings of a \mc analysis.
And second, it could well be that many more existing \mc analyses are tight.
} 
}

\section{The \smc technique}
\label{sec:smc}

Our \mrcsp algorithm illustrates that one can exploit separator-based branching to design a more efficient exponential-time algorithm.
However, {the \mrcsp problem and its algorithms} have certain features that make the analysis simpler than for other problems. In this section, we outline what additional complications could arise in the analysis and how to handle them. In the next section, we will illustrate this more general method.
First, the only branching rule of the \mrcsp algorithm, \Red{3}, produces {$r$ }instances with {exactly the same} constraint graph, and therefore the same measure.
As a consequence, the constraints needed to satisfy \eqref{eq:magain} all become linear.
{But, typically, the instances produced by a branching rule have different measures, which leads to convex constraints. }%
Second, the measure $\mu_r(R)$ depends only on the number of degree-3 vertices in
$R$.
This implies a discretized change in the measure for $R$ whenever $L$ and $R$ are swapped.
{If the measure attaches incommensurable weights to vertices of different
types (different degrees, perhaps), then} the change in measure
{resulting from}
swapping $L$ and $R$ could take values dense within a continuous domain.
Our {simple }measure for the \mcsp algorithm also implies that the initial
separator only needs to balance the \emph{number} of vertices in $L$ and $R$
instead of the \emph{measure} of $L$ and $R$,
which is what is needed more generally.
Finally, a general method is needed to combine the separator-based branching, which would typically be done for sparse instances, with the general case, where vertex degrees are arbitrary.

{In this section we propose a general method to exploit separator-based branchings in the analysis of an algorithm. The section is tailored to readers familiar with Measure and Conquer and would like to use the present method to design and analyze new algorithms. The will take into account all the
complications mentioned in the previous paragraph, and we will illustrate its use
in the next section to obtain two faster
algorithms for counting dominating sets.}
{The technique will apply} to recursive algorithms that label vertices of a graph, and where an instance can be decomposed into two independent subinstances when all the vertices of a separator have been labeled in a certain way.
{
}
Let $G=(V,E)$ be a graph and $\ell : V \rightarrow \cL$ be a labeling of its vertices by labels in the finite set $\cL$.
{(Partial labelings are handled by including a label whose interpretation is
``unlabeled''.) }%
For a subset of vertices $W\subseteq V$, denote by $\mu_r(W)$ and $\mu_s(W)$ two measures for the vertices in $W$ in the graph $G$ labeled by $\ell$.
The measure $\mu_r$ is used for the vertices on the right hand side of the separator and $\mu_s$ for the vertices in the separator.
Let $(L,S,R)$ be a separation of $G$.
Initially, we use the separation $(L,S,R) = (\emptyset, \emptyset, V)$.
We define the measure
\begin{align}
 \mu(L,S,R) &= \mu_s(S) + \mu_r(R) + \max \left( 0, B-\frac{\mu_r(R)-\mu_r(L)}{2} \right) \notag\\
 &\quad + (1+B) \cdot \log_{1+\epsilon} (\mu_r(R)+\mu_s(S)), \label{eq:measure}
\end{align}
where $\epsilon>0$ is a constant{ greater than $0$} that will be chosen small
enough to satisfy constraint \eqref{eq:sep} below, and $B$ is an arbitrary
constant greater than the maximum change {(increase or decrease) }in imbalance
in each transformation in the analysis, except the Separation transformation.
The \emph{imbalance} of an instance is $\mu_r(R) - \mu_r(L)$, and we assume, as previously, that
\begin{align}
 \mu_r(R) \ge \mu_r(L).
\end{align}
{It is important that $B$ be an absolute constant: although the imbalance of an instance may change arbitrarily in an execution of an algorithm, our value of $B$ is only constrained to be greater than the maximum change in imbalance that the analysis takes into account.\footnote{Typically, if the imbalance decreases by a very large number for a given branching, it is sufficient to replace this number by a large absolute constant without compromising the quality of the analysis.}}
{We need two more assumptions about the measure so that it will be possible to compute a balanced separator. Namely,}
we will assume that adding a vertex to $R$ {(and by symmetry, removing a vertex from $R$) }changes {the measure }$\mu_r(R)$ by at most {a constant.
The value of this constant is not crucial for the analysis; so we assume for
simplicity that it is at most }$B$ (adjusting {the value of }$B$ if necessary):
\begin{align}
 | \mu_r(R\cup\set{v}) - \mu_r(R) |  &\le B \qquad \text{for each } R\subseteq V \text{ and } v\in V. \label{eq:find-sep}
\end{align}
We also assume that $\mu_r(R)$ can be computed in time polynomial in $|V|$ for each $R\subseteq V$.

Let us now look more closely at the measure \eqref{eq:measure}.
The terms $\mu_s(S)$ and $\mu_r(R)$ naturally define measures for the vertices in $S$ and $R${, respectively}.
No term of the measure directly accounts for the vertices in $L$; we merely enforce that $\mu_r(R) \ge \mu_r(L)$.
The term $\max \left( 0, B-\frac{\mu_r(R)-\mu_r(L)}{2} \right)$ is a penalty term based on how balanced the instance is: the more balanced the instance, the larger the penalty term.
The penalty term has become continuous, varying from $0$ to $B$.
{The exact definition of the penalty term will become clearer when we discuss the change in measure when branching. }%
The final term{, $(1+B) \cdot \log_{1+\epsilon} (\mu_r(R)+\mu_s(S))$,} amortizes the increase in measure of at most $B$ due to the balance terms each time the instance is separated.

Let us now formulate some generic constraints that the measure should obey.
{The first one concerns the separation reduction.}

\redsec{Separation.}
We assume that an instance with a separation $(L,S,R)$ can be separated into two independent subinstances $(L,S,\emptyset)$ and $(\emptyset,S,R)$ when the labeling of $S$
allows it;
specifically, when all vertices in $S$ have been labeled by a subset $\cL_s
\subseteq \cL${ of so-called \emph{separation}
labels.
The labels are very algorithm-specific.}
{%
In \mcsp, it sufficed that no label in $S$ was ``unlabeled''.
A Dominating Set algorithm might have $\cL_s = \{i,o_d,o_R,o_L\}$, separating the instance when all vertices in the separator have been restricted to be either in the dominating set ($i$), not in the dominating set and already dominated ($o_d$), not in the dominating set and needing to be dominated by a vertex in $R$ ($o_R$), or not in the dominating set and needing to be dominated by a vertex in $L$ ($o_L$). }%
{The separation reduction applies when $\ell(s)\in \cL_s$ for each $s\in S$, which} arises in two cases.
The first is at the beginning of the algorithm when the graph has not been separated, which is represented by the trivial separation $(\emptyset, \emptyset, V)$.
The second is when our reductions have produced a separable instance{, typically with $L$ and $R$ (and possibly $S$) nonempty}.

Let $(L,S,R)$ be such that $\ell(s)\in \cL_s$ for each $s\in S$. The algorithm recursively solves the subinstances $(L,S,\emptyset)$ and $(\emptyset,S,R)$.
Let us focus on the instance $(\emptyset,S,R)$; the treatment of the other instance is symmetric.
After a cleanup phase, {where the algorithm applies some simplification rules}{ (for example, the vertices labelled $o_L$ and needing to be dominated by a vertex in $L$ can be removed in this subinstance)}, the next step is to compute a new separator of $S\cup R$.
This can be done in various ways, depending on the graph class. For example, polynomial-time computable balanced separators can be derived from upper bounds on the pathwidth of graphs with bounded maximum or average degree \cite{EdwardsM15,FominGSS09,Gaspers10}.
{
Let us assume that we are dealing with graphs of maximum degree at most 6.

\begin{lemma}\label{lem:sep}
Let $G=(V,E)$ be a graph of maximum degree at most $\Delta, 3\le \Delta \le 6$,
and $\mu$ be a measure as in~\eqref{eq:measure} satisfying \eqref{eq:find-sep} such that $\mu_r(R)$ can be computed in time polynomial in $|V|$ for each $R\subseteq V$.
A separation $(L,S,R)$ of $G$ can be computed in polynomial time such that
$|\mu_r(L)-\mu_r(R)| \le B$ and $|S|\le \alpha_{\Delta} |V| + o(|V|)$, where
$\alpha_3 = 1/6$, $\alpha_4 = 1/3$, $\alpha_5 = 13/30$, and $\alpha_6 = 23/45$.
\end{lemma}
\begin{proof}
 By \cite{FominGSS09}, the pathwidth of $G$ is at most $\alpha_{\Delta} |V| + o(|V|)$ and a path decomposition of that width can be computed in polynomial time.
 We view a path decomposition as a sequence of bags $(B_1, \dots, B_b)$ which are subsets of vertices such that for each edge of $G$, there is a bag containing both endpoints, and for each vertex of $G$, the bags containing this vertex form a non-empty consecutive subsequence. The width of a path decomposition is the maximum bag size minus one.
 We may assume that every two consecutive bags $B_i$, $B_{i+1}$ differ by exactly one vertex, otherwise we insert between $B_i$ and $B_{i+1}$ a sequence of bags where the vertices from $B_i \setminus B_{i+1}$ are removed one by one followed by a sequence of bags where the vertices of $B_{i+1} \setminus B_i$ are added one by one; this is the standard way to transform a path decomposition into a \emph{nice} path decomposition of the same width where the number of bags is polynomial in the number of vertices \cite{BodlaenderK96}.
 Note that each bag is a separator and a bag $B_i$ defines the separation $(L_i, B_i, R_i)$ with $L_i = (\bigcup_{j=1}^{i-1} B_j)\setminus B_i$ and $R_i = V \setminus (L_i\cup B_i)$.
 Since the first of these separations has $L_1=\emptyset$ and the last one has $R_b=\emptyset$, at least one of these separations has $|\mu_r(L_i)-\mu_r(R_i)| \le B$ by \eqref{eq:find-sep}.
 Finding such a bag can clearly be done in polynomial time.
\end{proof}
}
\noindent
After a balanced separator $(L',S',R')$ has been computed for $S\cup R$, the instance is solved recursively, and so is the instance $L\cup S$, separated into $(L'',S'',R'')$. Both solutions are then combined into a solution for the instance $L\cup S\cup R$.
Without loss of generality, assume $\mu(L',S',R') \ge \mu(L'',S'',R'')$.
Assuming that the separation and combination are done in polynomial time, the imposed constraint on the measure is
\begin{align*}
 2\cdot 2^{\mu_r(R') + \mu_s(S') + B + (1+B) \cdot \log_{1+\epsilon} (\mu_r(R')+\mu_s(S'))}
 \le
 2^{\mu_r(R) + \mu_s(S) + (1+B) \cdot \log_{1+\epsilon} (\mu_r(R) + \mu_s(S))} .
\end{align*}
To satisfy the constraint, it suffices to {satisfy the two constraints:
\begin{align*}
 \mu_r(R') + \mu_s(S') &\le \mu_r(R) + \mu_s(S)\\
 (1+B) + (1+B) \cdot \log_{1+\epsilon} (\mu_r(R')+\mu_s(S')) &\le (1+B) \cdot \log_{1+\epsilon} (\mu_r(R) + \mu_s(S))
\end{align*}
To satisfy both constraints, it is sufficient to }constrain that
\begin{align}
\mu_r(R) + \mu_s(S) \ge (1+\epsilon) (\mu_r(R') + \mu_s(S')). \label{eq:sep}
\end{align}
This is the only constraint involving the size of a separator.
It constrains that separating $(\emptyset,S,R)$ to $(L',S',R')$ should reduce
$\mu_r(R)+\mu_s(S)$ by a constant factor, namely $1+\epsilon$.

\redsec{Branching.}
{%
 The branching rules will apply to two kinds of instances: balanced and imbalanced instances. We say that an instance is \emph{balanced} if $0 \le \mu_r(R) - \mu_r(L) \le 2B$ and \emph{imbalanced} if $\mu_r(R) - \mu_r(L) > 2B$. For a balanced instance, the measure \eqref{eq:measure} is
 \begin{align*}
  \mu(L,S,R) &= \mu_s(S) + \frac{\mu_r(R)}{2} + \frac{\mu_r(L)}{2} + B + (1+B) \cdot \log_{1+\epsilon} (\mu_r(R)+\mu_s(S))
 \intertext{and for an imbalanced instance, it is}
  \mu(L,S,R) &= \mu_s(S) + \mu_r(R) + (1+B) \cdot \log_{1+\epsilon} (\mu_r(R)+\mu_s(S)).
 \end{align*}
}
A standard Measure and Conquer analysis would probably assign a weight to every vertex and use a measure of the form $\mu_r(L)+\mu_s(S)+\mu_r(R)$, while our measure does not assign weights to vertices in $L$.
Let us recall the optimistic assumptions in Section \ref{sec:intuition} where we outlined an analysis for \mtcsp under ideal conditions: after exhausting the separator $S$, we assumed half of the other vertices that were removed by the branching were in $L$ and the other half in $R$. Thus the instance decomposed into two equally-small connected components.
We will now present a condition under which $\mu_r(R)$ decreases at least by half as much as $\mu_r(L)+\mu_r(R)$.

Suppose a transformation taking $(L,S,R,\ell)$ to $(L',S',R',\ell')$ decreases $\mu_r(R)+\mu_r(L)$ by $d$.
Since the measure includes roughly (and at least) half of $\mu_r(R)+\mu_r(L)$, in an ideal case the measure $\mu_r(R) + \max \left( 0, B-\frac{\mu_r(R)-\mu_r(L)}{2} \right)$ decreases by $d/2$.
{We will now show that this} is indeed the case for our measure if the {transformation guarantees }the following condition:
\begin{align}
\text{If } \mu_r(R) - \mu_r(L) > B \text{, then } \mu_r(R) - \mu_r(R') \ge \mu_r(L) - \mu_r(L')\enspace.\label{eq:balance-condition}
\end{align}
{Recall that $B$ was defined to be greater than the change in imbalance in the analysis of each transformation, so $\mu_r(R) - \mu_r(L) \le B$ implies the instance remains balanced after a transformation.}
Condition \eqref{eq:balance-condition} is very natural, expressing that, if
the instance is imbalanced or risks becoming imbalanced we would like to make more progress on the large side.
{
	
 In the balanced case, $0 \le \mu_r(R) - \mu_r(L) \le 2B$.
 If $\mu_r(R) - \mu_r(L) > B$, then by \eqref{eq:balance-condition}, $\mu_r(R')-\mu_r(L') \le \mu_r(R) - \mu_r(L) \le 2B$,
 while if $\mu_r(R) - \mu_r(L) \le B$, then by the definition of $B$, $\mu_r(R') - \mu_r(L') \le 2B$.
 It follows that the resulting instance $(L',S',R',\ell')$ is balanced as well, and $\mu_r(R) + \max \left( 0, B-\frac{\mu_r(R)-\mu_r(L)}{2} \right)$ decreases by $(\mu_r(R)+\mu_r(L) - \mu_r(R')-\mu_r(L'))/2 = d/2$.

 In the imbalanced case, $\mu_r(R) - \mu_r(L) > 2B$. By condition \eqref{eq:balance-condition}, $\mu_r(R) + \max \left( 0, B-\frac{\mu_r(R)-\mu_r(L)}{2} \right)$ decreases by
 \begin{align*}
 \mu_r(R)-\mu_r(R') &\ge (\mu_r(R) - \mu_r(R') + \mu_r(L) - \mu_r(L'))/2\\ & 
 = d/2
 \intertext{if $(L',S',R',\ell')$ is also imbalanced, or by}
 \mu_r(R)- \left( \frac{\mu_r(R')}{2} + \frac{\mu_r(L')}{2} + B \right) &= \frac{\mu_r(R)}{2}- \left( B - \frac{\mu_r(R)}{2} \right) -\frac{\mu_r(R')}{2}-\frac{\mu_r(L')}{2}\\
 &\ge (\mu_r(R) + \mu_r(L) - \mu_r(R') - \mu_r(L'))/2\\ & 
 = d/2
 \end{align*}
 if $(L',S',R',\ell')$ is balanced.

}
Thus, if Condition~\eqref{eq:balance-condition} holds, {i.e., we can guarantee more progress on the large side if the instance is imbalanced or risks to become imbalanced, }then the analysis is at least as good as a non-separator based analysis, but with the additional improvement due to the separator branching.

\redsec{Integration into a standard \mc analysis.}
The \smc analysis will typically be used when the instance has become sufficiently sparse that one can guarantee that a small separator exists.
We can view the part
played by the \smc analysis as a subroutine and integrate it into any other \mc analysi. We only need guarantee that the measure of an instance does
not increase when transitioning to the subroutine.
{Formally, this is done with the following lemma that can be proved with a simple induction.
\begin{lemma}[\cite{Gaspers10,hybrid}] \label{lem:combinemeasureanalysis}
Let $A$ be an algorithm for a problem $P$,
$B$ be an algorithm for (special instances of) $P$,
$c \ge 0$ and $r>1$ be constants,
and $\mu(\cdot), \mu'(\cdot), \eta(\cdot)$ be measures
for the instances of $P$,
such that
for any input instance $I$, $\mu'(I) \le \mu(I)$, and
for any input instance $I$,
$A$ either solves $P$ on $I$ by invoking $B$ with running time $O(\eta(I)^{c+1}) r^{\mu'(I)}$,
or reduces $I$ to instances $I_1,\ldots,I_k$,
solves these recursively, and combines their solutions to solve~$I$,
using time $O(\eta(I)^{c})$ for the reduction and combination steps
(but not the recursive solves),
\begin{align}
(\forall i) \quad \eta(I_i) & \leq \eta(I)-1 \text{, and}
  \\
\sum_{i=1}^k r^{\mu(I_i)} & \leq r^{\mu(I)} .
\end{align}
Then $A$ solves any instance $I$
in time $O(\eta(I)^{c+1}) r^{\mu(I)}$.
\end{lemma}
}
{%
This completes the description of the \smc method.
We now illustrate its use to obtain improved algorithms for \textsc{\#Dominating Set}.
}

\section{Counting dominating sets}
\label{sec:cds}

{\textsc{Dominating Set} (\DS) and its variants are of central importance in exponential-time algorithms.
Let us denote the number of vertices of the input graph by $n$.
\DS can be solved in time $O(1.4864^n)$ using polynomial space \cite{Iwata11} and in time $O(1.4689^n)$ using exponential space \cite{Iwata11}.
All minimal dominating sets can be enumerated in time $O(1.7159^n)$ and polynomial space \cite{FominGPS08}.
Minimum Weighted Dominating Set can be solved in time $O(1.5535^n)$ and exponential space \cite{FominGSS09}.
Partial Dominating Set can be solved in time $O(1.5673^n)$ using polynomial space \cite{NederlofR10} and in time $O(1.5012^n)$ using exponential space \cite{NederlofR10}.
Other variants have been considered (see, e.g., \cite{Liedloff07,Rooij11}) and the problems have been studied on many graph classes (see, e.g., \cite{CouturierHHK13,CouturierLL15,GaspersKLT09,GolovachHK16,Krzywkowski13,Liedloff08,NederlofRD14}). }%
The \NDS problem is to compute, for a given graph $G$, the function
{$d: \set{0,\dots,n} \rightarrow \mathbb{N}$}
such that {for each $k\in \set{0,\dots,n}$, }$d(k)$ is the number of dominating sets of $G$ of size $k$.

{The} current fastest polynomial-space algorithm {for \NDS is a} $O(1.5673^n)${ time algorithm by van Rooij} \cite{Rooij10}.\footnote{Using exponential space, the problem can be solved in time $O(1.5002^n)$ \cite{NederlofRD14}.}
{%
 Like many \DS algorithms, it uses a construction due to Grandoni \cite{Grandoni06} to transform the instance into an instance for \textsc{Set Cover}.
 Given a multiset $\mathcal{S}$ of subsets of a universe $\mathcal{U}$, a collection of subsets $C \subseteq \mathcal{S}$ is a \emph{set cover} if $\bigcup_{c\in C} c = \mathcal{U}$.
 The transformation from \DS to \textsc{Set Cover} sets $\mathcal{U}:=V$ and adds the set $N_G[v]$ to $\mathcal{S}$ for each vertex $v$. Here, $N_G[v]$ denotes the closed neighborhood of $v$, that is, the set of vertices containing $v$ and all its neighbors.
}%
While many algorithms for {\DS and variants} rely on {this} transformation to \textsc{Set Cover}, the current fastest polynomial-space algorithm for subcubic graphs
{is {a simple} $\Ostar{2^{n/2}}$ time algorithm \cite{KneisMRR05} that works directly on the input graph.
}
{(It still outperforms all more recent algorithms on general graphs \cite{FominGK09,Iwata11,RooijB11} when their analysis is restricted to the subcubic case.)}%

In this section, we apply the \smc method
to design and analyze faster algorithms for \NDS for subcubic graphs and,
separately, for general graphs.
{%
For subcubic graphs {we design an algorithm working} directly on the input graph, where we can essentially reuse the Max $(3,2)$-CSP analysis of Section~\ref{sec:csp}.
For general graphs, our algorithm is based on the \textsc{Set Cover} translation and essentially just adds separation to \cite{Rooij10}.
}%

{
\subsection{Subcubic graphs}
\label{subsec:3cds}

\newcommand{\cubstr}{\Gamma}

We assume the input graph $G$ has maximum degree at most 3, and we denote by $\cubstr(G)$ the \emph{cubic structure} of $G$, the unique 3-regular graph obtained from $G$ by exhaustively contracting edges incident to vertices of degree at most $2$ and removing isolated vertices.
Our algorithm will branch on the vertices of $\cubstr(G)$ in the same way as the $\mrcsp$ algorithm and label the vertices of $G$ so as to remember which vertices still need to be dominated, which ones should not be dominated, etc. Once the number of vertices in $\cubstr(G)$ is upper bounded by a constant, the problem can be solved in polynomial time by tree-decomposition methods.
We will first describe the labeling, then the base case where the number of vertices of $\cubstr(G)$ is bounded by a constant, and then the branching strategy.

To solve the \NDS problem, the algorithm applies labels $U$, $N$, 
and $C$ to vertices.
We denote by $\ell(v) \in \set{U,N,C}$ the label of vertex $v$.
The algorithm counts the number of sets $D$ for each size $k$, $0\le k\le |V|$, with the following restrictions according to the label of each vertex $v$:
\begin{itemize}
	\item ``unlabeled'' $U$: $v$ needs to be dominated by $D$;
	\item ``not in'' $N$: $v$ is not in $D$, but needs to be dominated by $D$;
	\item ``covered'' $C$: there is no restriction on $v$, i.e., $v$ may or may not be in $D$ and $v$ does not need to be dominated by $D$.
\end{itemize}
Initially, all vertices are labeled $U$.
There is no label for vertices that do not need to be dominated and do not belong to a dominating set, since they are simply deleted from the graph.
Vertices that are added to the dominating set are also removed from the graph and their neighbors are updated to reflect that they do not need to be dominated any more: neighbors labeled $U$ or $C$ get label $C$, and neighbors labeled $N$ are deleted.
When we speak of a dominating set of a labeled graph $G$, we mean a vertex set containing vertices labeled $U$ and $C$ that dominates all vertices labeled $U$ and $N$.

If $|V(\cubstr(G))|\le A$ for a large enough constant $A\ge 2$, the instance is solved in polynomial time as follows.
First, a tree decomposition of $G$ of width at most $A$ is computed.
To do this, we start from a trivial tree decomposition of $\cubstr(A)$ that has just one bag with all the vertices of $\cubstr(A)$.
In polynomial time, one can augment that tree decomposition with the vertices from $V(G)\setminus V(\cubstr(A))$ and make sure that its width does not exceed $A$~\cite{Bodlaender98}.
Intuitively, each degree-0 vertex can get its own private bag and if a vertex $v$ of degree at most 2 was contracted away, choose a bag containing its neighbor(s) and add a neighboring bag containing $v$ and its neighbor(s).
Now, use a tree decomposition based algorithm, such as \cite{RooijBR09}, to solve the instance in polynomial time.
The latter treats unlabeled graphs but needs only minor adaptations for labeled graphs.

We will now describe the algorithm for the case when $|V(\cubstr(G))|> A$.
\newcommand{\cds}{\textrm{\#ds}}
The algorithm will satisfy the invariant that degree-3 vertices in $G$ are labeled $U$.
It will compute an $(n+1)$-size vector $\cds_G$ such that the labeled graph $G$ has $\cds_G[i]$ dominating sets of size $i$, $0\le i\le n$.

Analogous to the algorithm in Section~\ref{sec:csp}, we now define a \Red{3}.
(Analogs of \Red{0--2} are not needed here.)

\begin{description}
\item[\Red{3}] Let $x$ be a vertex of degree 3. Thus, $\ell(x)=U$. There are three subinstances, $G_{\text{in}}, G_{\text{opt}}$, and $G_{\text{forb}}$ (thought of as ``in'', ``optional'', and ``forbidden''). They are obtained from $G$ as follows.
\begin{itemize}
 \item In $G_{\text{in}}$, we add $x$ to the dominating set. The vertex $x$ is deleted from the graph, its neighbors labeled $U$ are relabeled $C$, and its neighbors labeled $N$ are deleted.
 \item In $G_{\text{opt}}$, we prevent $x$ from being in the dominating sets and we remove the requirement that it needs to be dominated. The vertex $x$ is removed from the graph.
 \item In $G_{\text{forb}}$, we forbid that $x$ is dominated, that is we are interested in the number of dominating sets of $G-x$ that do not contain any vertex from $N_G(x)$. Delete the vertices in $N_G(x)$ labeled $C$, assign label $N$ to the remaining vertices of $N_G(x)$, and delete $x$.
\end{itemize}
The algorithm computes the number of dominating sets of size $i$ of $G$ by setting
\begin{align*}
\cds(G)[i] = \cds(G_{\text{in}})[i-1] + \cds(G_{\text{opt}})[i] - \cds(G_{\text{forb}})[i].
\end{align*}
The algorithm returns the vector $\cds$.
\end{description}

To see that \Red{3} correctly computes the number of dominating sets of size $i$ of $G$, observe that
these sets that contain the vertex $x$ are counted in $\cds(G_{\text{in}})[i-1]$.
The dominating sets of size $i$ of $G$ that do not contain $x$ are the dominating sets of size $i$ of $G-x$, counted in $\cds(G_{\text{opt}})[i]$, except for those that do not dominate $x$ (counted in $-\, \cds(G_{\text{forb}})[i]$).

We observe that \Red{3} has at least the same effect on $\cubstr(G)$ as \Red{3} in \mcsp algorithm of Section~\ref{sec:csp} has on the constraint graph: the vertex $x$ is deleted, and all vertices with degree at most $2$ in $\cubstr(G-x)$ are contracted or, even better, deleted.
\Reds{0-2} are replaced by operations doing nothing. This has the same effect on $\cubstr(G)$ as \Reds{0-2} in \mcsp algorithm of Section~\ref{sec:csp} has on the constraint graph, which either becomes cubic again or vanishes.
Thus, the algorithm will select the pivot vertex for branching in exactly the same way as the algorithm in Section~\ref{sec:csp}.
We note that the running time analysis of Section~\ref{sec:csp} was based only on the measure of the constraint graph.
Since degree-3 vertices are handled by a 3-way branching, the running time
analysis for Max $(3,2)$-CSP applies for \NDS as well, giving a running time of $3^{n/5+o(n)} = O(1.2458^n)$ for subcubic graphs.
This improves on the previously fastest polynomial-space algorithm with running time $\Ostar{2^{n/2}} = O(1.4143^n)$ \cite{KneisMRR05}.

\begin{theorem}
	The described algorithm solves \NDS
		in $3^{n/5+o(n)}$ time and polynomial space on subcubic graphs.
\end{theorem}

There are two novel ideas that make this improved running time possible, and we will briefly discuss how much
each idea contributed to the improvement.
The first one is branching on vertices in a small separator,
and the second one is our new 3-way branching.
Our 3-way branching is inspired by the inclusion/exclusion branching of \cite{RooijND09}, but
to the best of our knowledge, it has not been used in this form before.
The algorithm of \cite{KneisMRR05} used a 4-way branching that gave a running time of $\Ostar{4^{n/4}} = \Ostar{2^{n/2}}=O(1.4143^n)$.
If we use their analysis, or the analysis of \cite{ScottS07}, together with our 3-way branching, one obtains a running time bound of $\Ostar{3^{n/4}} = O(1.3161^n)$.
The further improvement to $3^{n/5+o(n)}=O(1.2458^n)$ is due solely to the
separator branching.
}%
\subsection{General graphs}
\label{subsec:cds}

As another application of our \smc technique, we show that by changing the preference for the pivot vertex in the subcubic subproblem, taking advantage of a separator, the running time of van Rooij's algorithm \cite{Rooij10} improves from $O(1.5673^n)$ to $O(1.5183^n)$.

\begin{algorithm}[tbp]
	\caption{\label{algo:vanRooij}\#SC-vR($I,A$) -- van Rooij's algorithm for \textsc{\#Set Cover}.}
	\KwIn{The incidence graph $I=(\mathcal{S} \cup \mathcal{U}, E)$ of $(\mathcal{U},\mathcal{S})$ and a set of annotated vertices $A$}
	\KwOut{A list containing the number of set covers of $\mathcal{S}$ of each cardinality $\kappa$}
	\SetKwComment{LabRule}{}{}
	\BlankLine
	\If{there exists a vertex $v\in V(I)\setminus A$ of degree at most one in $I-A$} {
		\Return \#SC-vR($I,A \cup \set{v}$)
	}
	\ElseIf{there exist two vertices $v_1, v_2 \in V(I)\setminus A$ both of degree two in $I-A$ that have the same two neighbors} {
		\Return \#SC-vR($I,A \cup \set{v_1}$)
	}
	\Else {
		Let $X\in \mathcal{S}\setminus A$ be a set vertex of maximum degree in $I-A$\\
		Let $e\in \mathcal{U}\setminus A$ be an element vertex of maximum degree in $I-A$\\
		\If{$d_{I-A}(X)\le 2$ \textbf{and} $d_{I-A}(e)\le 2$} {
			\Return \#SC-DP($I$)\\ \nllabel{ln:before}
		}
		\ElseIf{$d_{I-A}(X) > d_{I-A}(e)$} { \nllabel{ln:after}
			Let $L_{\textrm{take}} = $ \#SC-vR($I-N_I[X], A \setminus N_I(X)$) and increase all cardinalities by one\\
			Let $L_{\textrm{discard}} = $ \#SC-vR($I-X, A$)\\
			\Return $L_{\textrm{take}} + L_{\textrm{discard}}$
		}
		\Else {
			Let $L_{\textrm{optional}} = $ \#SC-vR($I-e, A$)\\
			Let $L_{\textrm{forbidden}} = $ \#SC-vR($I-N_I[e], A\setminus N_I(e)$)\\
			\Return $L_{\textrm{optional}} - L_{\textrm{forbidden}}$
		}
	}
\end{algorithm}

Let us recall van Rooij's algorithm \#SC-vR \cite{Rooij10} (see
Algorithm~\ref{algo:vanRooij}). It first reduces the \NDS problem to the \textsc{\#Set Cover} problem, where each vertex corresponds to an element of the universe $\mathcal{U}$ and the closed neighborhood of each vertex corresponds to a set $X\in \mathcal{S}$.
Then, there is a cardinality-preserving bijection between dominating sets of the graph and set covers of the instance $(\mathcal{U},\mathcal{S})$.
The \emph{incidence graph} of $(\mathcal{U},\mathcal{S})$ is the bipartite graph $(\mathcal{S} \cup \mathcal{U},E)$ which has an edge $(X,e)\in E$ between a set $X\in \mathcal{S}$ and an element $e\in \mathcal{U}$ if and only if $e\in X$.
In the \textsc{\#Set Cover} problem, the input is an incidence graph of an instance $(\mathcal{U},\mathcal{S})$, and the task is to compute for each size $\kappa$, $0\le \kappa \le |\mathcal{S}|$, the number of set covers of size $\kappa$.
As an additional input, the algorithm has a set $A$ of annotated vertices that is initially empty.
When the algorithm finds a vertex of degree at most 1 or a vertex of degree 2 that has the same neighborhood as some other vertex, it annotates this vertex, which means the vertex is effectively removed from the instance, and is left to a polynomial-time dynamic programming algorithm, \#SC-DP, that handles instances of maximum degree at most 2 plus annotated vertices.

\begin{algorithm}[tbp]
	\caption{\label{algo:cds}\#SC($I,A,(L,S,R)$) -- separator-based \textsc{\#Set Cover} algorithm.
		(Lines shown in grey are essentially unchanged from Algorithm~\ref{algo:vanRooij}.)
	}
	\KwIn{The incidence graph $I=(\mathcal{S} \cup \mathcal{U}, E)$ of
		$(\mathcal{U},\mathcal{S})$, a set of annotated vertices $A$, and a separation $(L,S,R)$ of $I-A$}
	\KwOut{A list containing the number of set covers of $\mathcal{S}$ of each cardinality $\kappa$}
	\SetKwComment{LabRule}{}{}
	\BlankLine
	\If{$I$ has a connected component $X\subsetneq V(I)$} {
		Let $L_{X} = $ \#SC($I[X], A\cap X$, $(\emptyset,\emptyset,X\setminus A)$)\\
		Let $L_{\overline{X}} = $ \#SC($I - X, A\setminus X, (\emptyset,\emptyset,V(I)\setminus (X\cup A)$)\\
		\Return $L$ where $L[i] = \sum_{j=0}^i L_{X}[j] \cdot L_{\overline{X}}[i-j]$
	}
	\color{black!60}
	\ElseIf{there exists a vertex $v\in V(I)\setminus A$ of degree at most one in $I-A$} {
		\Return \#SC($I,A \cup \set{v}, (L\setminus\set{v}, S\setminus\set{v},R\setminus\set{v})$) \nllabel{nl:deg1}
	}
	\ElseIf{there exist two vertices $v_1, v_2 \in V(I)\setminus A$ both of degree two in $I-A$ that have the same two neighbors} {
		\Return \#SC($I,A \cup \set{v_1}, (L\setminus\set{v_1}, S\setminus\set{v_1},R\setminus\set{v_1})$)
	}
	\Else {
		Let $X\in \mathcal{S}\setminus A$ and $e\in \mathcal{U}\setminus A$ be a set vertex and an element vertex, each of maximum degree in $I-A$\\
		\If{$d_{I-A}(X)\le 2$ \textbf{and} $d_{I-A}(e)\le 2$} {
			\Return \#SC-DP($I$)
		}
		\color{black}
		\ElseIf{$d_{I-A}(X)\le 3$ \textbf{and} $d_{I-A}(e)\le 3$} {
			\Return \#3SC($I,A,(L,S,R)$)
		}
		\color{black!60}
		\ElseIf{$d_{I-A}(X) > d_{I-A}(e)$} {
			\Return Branch($I,A,(L,S,R),X$)
		}
		\Else {
			\Return Branch($I,A,(L,S,R),e$)
		}
		\color{black}
	}
\end{algorithm}

Van Rooij showed that this algorithm has running time $O(1.5673^n)$.
We will now modify it to obtain Algorithm \#SC (see Algorithm~\ref{algo:cds}), which takes advantage of small separators in subcubic graphs.
Then, we will show that this modification improves the running time to $O(1.5183^n)$.
Algorithmically, the change is simple. The parts that remain essentially
unchanged (except that the separation is passed along in the recursive calls)
are written in {\color{black!60}gray} in Algorithm~\ref{algo:cds}.
For convenience, we also define as a function what it means to branch on a vertex (see Function \ref{algo:branch})
since we use it several times.
\begin{function}[tbp]
	\caption{Branch($I,A,{\left( L,S,R\right)}, v$)} \label{algo:branch}
	\KwIn{The incidence graph $I=(\mathcal{S} \cup \mathcal{U}, E)$ of
		$(\mathcal{U},\mathcal{S})$, a set of annotated vertices $A$, and a separation $(L,S,R)$ of $I-A$, and a vertex $v\in (\mathcal{S} \cup \mathcal{U})\setminus A$ to branch on}
	\KwOut{A list containing the number of set covers of $\mathcal{S}$ of each cardinality $\kappa$}
	\SetKwComment{LabRule}{}{}
	\BlankLine
	\If{$v\in \mathcal{S}$} {
		Let $L_{\textrm{discard}} = $ \#SC($I-v, A,(L\setminus\set{v},S\setminus\set{v},R\setminus\set{v})$)\\
		Let $L_{\textrm{take}} = $ \#SC($I-N_I[v], A \setminus N_I[v], (L\setminus N[v],S\setminus N[v], R\setminus N[v])$) and increase all cardinalities by one\\
		\Return $L_{\textrm{take}} + L_{\textrm{discard}}$ \nllabel{nl:branch2e}
	}
	\Else{
		Let $L_{\textrm{optional}} = $ \#SC($I-v, A,(L\setminus\set{v},S\setminus \set{v},R\setminus \set{v})$)\\
		Let $L_{\textrm{forbidden}} = $ \#SC($I-N_I[v], A\setminus N_I[v], (L\setminus N[v],S\setminus N[v],R\setminus N[v])$)\\
		\Return $L_{\textrm{optional}} - L_{\textrm{forbidden}}$ \nllabel{nl:branch1e}
	}
\end{function}
When the algorithm reaches an instance with maximum degree at most $3$, it computes a separation of the incidence graph, minus the annotated vertices, and prefers to branch on vertices in the separator.
To do this, we add to the input a separation $(L,S,R)$ of $I-A$.
Initially, this separation is $(\emptyset,\emptyset,V(I)-A)$; it is really only used when the maximum degree of $I-A$ decreases to 3.
We note that a separation $(L,S,R)$ for $I-A$ can easily be transformed into a separation $(L',S',R')$ for $I$.
For this, we assign each vertex from $A$ to one of its neighbors in $V(I)\setminus A$ when it is annotated (or to an arbitrary vertex if it has degree 0).
Now, start with $(L',S',R')=(L,S,R)$.
A vertex from $A$ is added to $L'$, $S'$, or $R'$ if it is assigned to a vertex that is in $L'$, $S'$, or $R'$, respectively.
The vertices in $A$ will not affect the measure, defined momentarily.
Our algorithm also explicitly handles connected components when the graph is not connected, and it calls a subroutine, \#3SC (Algorithm~\ref{algo:3cds}), when the maximum degree is 3.

\begin{algorithm}[tbp]
\caption{\label{algo:3cds}\#3SC($I,A,(L,S,R)$) -- separator-based \textsc{\#Set Cover} algorithm for subcubic instances.}
\KwIn{The incidence graph $I=(\mathcal{S} \cup \mathcal{U}, E)$ of
$(\mathcal{U},\mathcal{S})$, a set of annotated vertices $A$, and a separation $(L,S,R)$ of $I-A$}
\KwOut{A list containing the number of set covers of $\mathcal{S}$ of each cardinality $\kappa$}
\SetKwComment{LabRule}{}{}
\BlankLine
     \If{$S=\emptyset$} {
       Compute a balanced separation $(L,S,R)$ of $I-A$ with respect to the measure $\mu_r$ using Lemma~\ref{lem:sep} \nllabel{nl:sep}
     }
     \If{$\mu_r(L) > \mu_r(R)$} {
       Swap $L$ and $R$
     }
     \If{there exists a vertex $s\in S$ with no neighbor in $L$} { \nllabel{nl:no-nb-L-cond}
      \Return \#SC($I,A, (L, S\setminus\set{s},R\cup\set{s})$) \nllabel{nl:no-nb-L}
     }
     \ElseIf{there exists a vertex $s\in S$ with no neighbor in $R$} { \nllabel{nl:no-nb-R-cond}
      \Return \#SC($I,A, (L\cup\set{s}, S\setminus\set{s},R)$) \nllabel{nl:no-nb-R}
     }
     \ElseIf{there exists a vertex $s\in S$ with $d_{I-A}(s)=2$} { \nllabel{nl:deg2-S}
       \If{$s$ has a degree-2 neighbor in $I-A$}{
       	Let $r$ be the first degree-3 vertex or vertex from $S$ encountered when moving from $s$ to the right along a path of degree-2 vertices in $I-A$\\
       	\Return Branch($I,A,(L,S,R),r$)
       	\nllabel{nl:deg2-S-path}
	   }
       \ElseIf{$(L,S,R)$ is balanced} {
        \Return \#SC($I,A, (L\setminus \set{l}, (S\setminus\set{s}) \cup\set{l}, R\cup \set{s})$) where $\set{l} = L\cap N_{I-A}(s)$ \nllabel{nl:deg2-S-bal}
       }
       \Else{
       	\Return \#SC($I,A, (L\cup \set{s}, (S\setminus\set{s}) \cup\set{r}, R\setminus \set{r})$) where $\set{r} = R\cap N_{I-A}(s)$ \nllabel{nl:deg2-S-imbal}
       }
     }
     \ElseIf{$\mu_r(R)-\mu_r(L)>B$ and there exists a vertex $s\in S$ with two neighbors in $L$ and one neighbor $r$ in $R$ that has degree $3$ in $I-A$} { \nllabel{nl:imbal-2L-cond}
      \Return \#SC($I,A, (L\cup \set{s}, (S\cup \set{r})\setminus \set{s}, R\setminus \set{r})$) \nllabel{nl:imbal-2L}
     }
     \ElseIf{$\mu_r(R)-\mu_r(L)>B$ and there exists a vertex $s\in S$ with two neighbors in $L$ and one neighbor $r$ in $R$ with $N_{I-A}(r)=\set{s,s'}$ for some $s'\in S$} { \nllabel{nl:imbal-2LS-cond}
      \Return \#SC($I,A, (L\cup \set{s,r}, S\setminus \set{s}, R\setminus \set{r})$) \nllabel{nl:imbal-2LS}
     }
     \ElseIf{there exists an element $s\in \mathcal{U} \cap S$} { \nllabel{nl:branch1s}
      \Return Branch($I,A,(L,S,R),s$)
     }
     \Else {
      Let $s\in \mathcal{S} \cap S$ be a set in $S$\\ \nllabel{nl:branch2s}
      \Return Branch($I,A,(L,S,R),s$) \nllabel{nl:branch2e}
     }
\end{algorithm}

\newcommand{\welt}{w_{\textrm{elt}}}
\newcommand{\wset}{w_{\textrm{set}}}
\newcommand{\wright}{w_{\textrm{right}}}
\newcommand{\wsep}{w_{\textrm{sep}}}

We are now ready to upper bound the running time of Algorithm \#SC (Algorithm~\ref{algo:cds}).
For instances with maximum degree at least $4$, we use the same measure as van Rooij,
\begin{align*}
 \mu_4 = \sum_{e\in \mathcal{U}\setminus A} \welt(d_{I-A}(e)) + \sum_{X\in \mathcal{S}\setminus A} \wset(d_{I-A}(X))\enspace.
\end{align*}
Here, $\welt(\cdot)$ and $\wset(\cdot)$ are non-negative real functions of the vertex degrees, which will be determined later.
For instances with maximum degree $3$, we use the following measure,
\begin{align*}
 \mu_3 &= \mu_s(S) + \mu_r(R) + \max \left( 0, B-\frac{\mu_r(R)-\mu_r(L)}{2} \right) + (1+B) \cdot \log_{1+\epsilon} (\mu_r(R)+\mu_s(S)),\\
\intertext{where}
 \mu_s(S) &= \sum_{s\in S} \wsep(d_{I-A}(s)),\\
 \mu_r(R) &= \sum_{r\in R} \wright(d_{I-A}(r)), \text{and}\\
 B &= 6 \cdot \wright(3).
\end{align*}
Again, $\wsep(\cdot)$ and $\wright(\cdot)$ are non-negative real functions of the vertex degrees, which will be determined later. We commonly refer to $\welt(i)$, $\wset(i)$, $\wsep(i)$, $\wright(i)$ as the weight of an element vertex, a set vertex, a separator vertex and a right vertex, respectively, of degree $i$.

To combine the analysis for subcubic instances and instances with maximum degree at least $4$ using Lemma~\ref{lem:combinemeasureanalysis}, we impose the following constraints, guaranteeing that $\mu_4\ge \mu_3$ when $\mu_4 = \omega(1)$. (When $\mu_4 = O(1)$, the algorithm will take constant time.)
\begin{align}
 \welt(3) &> \wright(3)                & \wset(3) &> \wright(3)\\
 \welt(i) &\ge \wright(i), \text{ and} & \wset(i) &\ge \wright(i), & i\in\set{0,1,2}.
\end{align}
For instances with degree at least $4$, we obtain exactly the same constraints on the measure as in \cite{Rooij10}.
Denoting $\Delta \welt(i) = \welt(i) - \welt(i-1)$ and $\Delta \wset(i) = \wset(i) - \wset(i-1)$, these constraints are
\begin{align}
 \welt(0) &= \welt(1) = 0,                 & \wset(0) &= \wset(1) = 0, \label{eq:w01}\\
 \Delta \welt(i) &\ge 0,                   & \Delta \wset(i) &\ge 0 &\text{for all } i\ge 2,\\
 \Delta \welt(i) &\ge \Delta \welt(i+1),   & \Delta \wset(i) &\ge \Delta \wset(i+1) &\text{for all } i\ge 2,\\
 2\cdot \Delta \welt(3) &\le \welt(2), \text{ and} & 2 \cdot \Delta \wset(4) &\le \wset(2). \label{eq:deg-dec-N2}
\end{align}
For branching on a vertex with degree $d\ge 4$ with $r_i$ neighbors of degree $i$ in $I-A$, where $\sum_{i=2}^{d} r_i = d$, we have the following constraints,
\begin{align}
 2^{-\wset(d) - \sum_{i=2}^{\infty} r_i \cdot \welt(i) - \Delta \wset(d) \cdot \sum_{i=2}^{\infty} r_i \cdot (i-1)}
 + 2^{-\wset(d) - \sum_{i=2}^{\infty} r_i \cdot \Delta \welt(i)}
 &\le 1 \label{eq:ds4set} 
\intertext{if the vertex is a set vertex, and}
 2^{-\welt(d) - \sum_{i=2}^{\infty} r_i \cdot \wset(i) - \Delta \welt(d) \cdot \sum_{i=2}^{\infty} r_i \cdot (i-1)}
 + 2^{-\welt(d) - \sum_{i=2}^{\infty} r_i \cdot \Delta \wset(i)}
 &\le 1 \label{eq:ds4elt} 
\end{align}
if it is an element vertex. We note that the constraints \eqref{eq:ds4set} have $r_d=0$ since the algorithm prefers to branch on elements when there is both an element and a set of maximum degree at least 4.
Constraints \eqref{eq:w01}--\eqref{eq:ds4elt} are directly taken from \cite{Rooij10}.

\newcommand{\degdec}{\delta_{\textrm{deg-dec}}}
For instances where $I-A$ has maximum degree $3$, let us recall (from Section \ref{sec:smc}) that
an instance is \emph{balanced} if $0 \le \mu_r(R) - \mu_r(L) \le 2B$ and \emph{imbalanced} if $\mu_r(R) - \mu_r(L) > 2B$.
For convenience, we define the minimum decrease in measure due to a decrease in the degree in $I-A$ of a vertex from $S\cup R$ of degree $2$ or $3$ in a balanced or imbalanced instance:
\begin{align}
 \degdec = \min_{i\in\set{2,3}} \left\{ \wsep(i)-\wsep(i-1) , \frac{\wright(i)-\wright(i-1)}{2} \right\}. \label{eq:def-deltadec}
\end{align}
To make sure that annotating vertices does not increase the measure, we require that the decrease of the degree of a vertex (half-edge deletion) does not increase the measure:
\begin{align}
 \degdec \ge 0. \label{eq:deltadec-pos}
\end{align}
For computing a new separator in line \ref{nl:sep}, constraint \eqref{eq:sep} becomes
\begin{align}
 \wsep(3)/6 + 5/12\cdot \wright(3) &< \wright(3), \text{ or} \notag\\
 \wsep(3) &< 7/2\cdot \wright(3).
\end{align}
For dragging a separator vertex into $R$ if it has no neighbor in $L$ (line \ref{nl:no-nb-L}), imbalanced instances are most constraining:
\begin{align}
 -\wsep(d) + \wright(d) &\le 0, & 2\le d\le 3. \label{eq:no-nb-L}
\end{align}
For dragging a separator vertex into $L$ if it has no neighbor in $R$ (line \ref{nl:no-nb-R}), balanced instances are most constraining, imposing the constraints
\begin{align*}
 -\wsep(d) + 1/2\cdot \wright(d) &\le 0, & 2\le d\le 3,
\end{align*}
which are no more constraining than \eqref{eq:no-nb-L}.

If there is a vertex $s\in S$ with $d_{I-A}(s)=2$ (line~\ref{nl:deg2-S}), it has one neighbor in $L\setminus A$ and one neighbor in $R\setminus A$, otherwise a previous case would have applied.
The algorithm distinguishes three cases.
In the first case, $s$ has a neighbor $x\in L\cup R$ that has degree 2 in $I-A$. The algorithm branches on $r$, which is the first vertex we encounter when moving from $s$ to the right in $I-A$, that is not a degree-2 vertex in $R$. So, $r$ is either in $S$ or it is a degree-3 vertex in $R$.
If $r\in S$, then branching on $r$ removes the vertex $r$ from the instance, and the vertex $s$ is also removed from $I-A$, either directly, or by triggering the degree-1 rule in line \ref{nl:deg1} of Algorithm \#SC.
Thus, the measure decreases by at least $2\cdot\wsep(2)$ in both branches. Thus, we will constrain that
\begin{align}
 2\cdot 2^{-2\cdot \wsep(2)} &\le 1 \notag\\
 \intertext{or, equivalently, that}
 2\cdot \wsep(2) &\ge 1.
\end{align}
We could weaken the constraint further by analyzing the effect on the measure of the other neighbors of $s$ and $r$, but the constraint will not turn out to be tight.
If $r\notin S$, then $r\in R$ and $r$ has degree 3 in $I-A$. After branching on $r$, we have a measure decrease of at least $\wright(3)/2$ because $r$ is removed ($r$ contributes $\wright(3)$ to the measure in the imbalanced case but only $\wright(3)/2$ in the balanced case). In addition, we have a measure decrease of $\wsep(2)$ because $s$ is removed, a decrease of $\wright(2)/2$ because $x$ is removed, and a decrease of at least $2\cdot \degdec$ because two more neighbors of $r$ have their degree reduced or are deleted. We note that $x$ could be in $L$ and its removal decreases $\mu_r(L)$, but we account for a decrease of $\mu_r(R)$ that is at least as large since $\wright(3)\ge \wright(2)$ due to constraints \eqref{eq:def-deltadec} and \eqref{eq:deltadec-pos}. Since Condition \eqref{eq:balance-condition} holds, the decrease of $\mu$ due to vertices in $R\cup L$ is therefore at least half the decrease of $\mu_r(L)+\mu_r(R)$.
For this subcase, we constrain that
\begin{align}
2\cdot 2^{-\wsep(2)-2\cdot \degdec -(\wright(2)+\wright(3))/2} &\le 1 \notag\\
\intertext{or, equivalently, that}
\wsep(2)+2\cdot \degdec +(\wright(2)+\wright(3))/2 &\ge 1.
\end{align}
In the second case, $s$ has two neighbors with degree 3 and the instance is balanced.
The algorithm ``drags'' $s$ to the right, that is, $s$ is moved from $S$ to $R$ and its neighbor in $L$ is moved to $S$.
We obtain the constraint
\begin{align}
 -\wsep(2) + \wsep(3) + 1/2\cdot (\wright(2)-\wright(3)) \le 0.
\end{align}
In the third case, $s$ has two neighbors with degree 3 and the instance is imbalanced.
The algorithm drags $s$ to the left, which leads to the weaker constraint
\begin{align*}
-\wsep(2) + \wsep(3) -\wright(3) \le 0.
\end{align*}

For the operation in line \ref{nl:imbal-2L}, where $\mu_r(R)-\mu_r(L)>B$ and some vertex $s\in S$ has two neighbors in $L$ and one degree-3 neighbor in $R$, the vertex $s$ is dragged to the left.
In the worst case, the resulting instance is balanced, imposing the following constraints on the measure,
\begin{align*}
 -\wsep(3)+\wsep(3) +1/2\cdot(\wright(3)-\wright(3)) &\le 0,
\end{align*}
which always holds.

For the operation in line \ref{nl:imbal-2LS}, where $\mu_r(R)-\mu_r(L)>B$ and some vertex $s\in S$ has two neighbors in $L$ and one degree-2 neighbor $r$ in $R$ that has another neighbor $s'$ in $S$, the vertices $s$ and $r$ are dragged to the left.
In the worst case, the resulting instance is balanced, imposing the following constraints on the measure,
\begin{align*}
 -\wsep(3) +1/2\cdot\wright(3) &\le 0,
\end{align*}
which holds due to \eqref{eq:no-nb-L}.

For the constraints of the branching steps, we will take into account the decrease of vertex degrees in the second neighborhood of a set that we add to the set cover and of an element that we forbid to cover.
The analysis will be simplified by imposing the following constraints,
\begin{align}
 2 \cdot \Delta \wsep(3) &\le \wsep(2) & 2\cdot \Delta \wright(3) &\le \wright(2).
\end{align}
They will enable us to account for a decrease in measure of $\degdec$ for each edge that leaves the second neighborhood, even when two of these edges are incident to one vertex of degree 2.
We cannot relax these equations for sets as in \eqref{eq:deg-dec-N2}: When $I-A$ has maximum degree 3 and it contains an element of degree 3, then van Rooij's algorithm prefers to branch on a degree-3 element, whereas the new algorithm cannot guarantee to find a degree-3 element in $S$.

For the two branching rules (lines \ref{nl:branch1s}--\ref{nl:branch1e} and \ref{nl:branch2s}--\ref{nl:branch2e}), the algorithm selects a separator vertex $s\in S$, which has degree 3 in $I-A$ due to line \ref{nl:deg2-S}.
In the first branch, $s$ is deleted from the graph, and in the second branch, $N_{I}[s]$ is deleted.

First, consider the case where $s$ has a neighbor $s'$ in $S$.
By lines \ref{nl:no-nb-L-cond} and \ref{nl:no-nb-R-cond}, $N_{I-A}(s) = \set{l,s',r}$ and $N_{I-A}(s') = \set{l',s,r'}$ with $\set{l,l'}\subseteq L$ and $\set{r,r'}\subseteq R$.
Since $I$ is bipartite, we have that $l\ne l'$ and $r\ne r'$.
In the balanced case, the branch where only $s$ is deleted has a decrease in measure of
$\wsep(3)$ for $s$, $\Delta \wsep(3)$ for $s'$, $1/2\cdot \Delta\wright(d_{I-A}(r))$ for $r$, and $1/2\cdot \Delta\wright(d_{I-A}(l))$ for $l$, and the branch where $N_{I-A}[s]$ is deleted has a decrease in measure of
$2\wsep(3)$ for $s$ and $s'$, $1/2\cdot \wright(d_{I-A}(r))$ for $r$, $1/2\cdot \wright(d_{I-A}(l))$ for $l$, and
$(d_{I-A}(r)+d_{I-A}(l))\cdot \degdec$ for the vertices in $N_{I-A}(N_{I-A}[s])$.
Thus, we get the following constraints for the balanced case,
\begin{align}
 &2^{-\wsep(3)-\Delta\wsep(3)-1/2\cdot(\Delta\wright(d_r)+\Delta\wright(d_l))} \notag\\
+&2^{-2\wsep(3)-1/2\cdot(\wright(d_r)+\wright(d_l))-(d_r+d_l)\cdot \degdec} \notag\\
\le& 1, & 2\le d_l,d_r\le 3.
\end{align}
In the imbalanced case, we obtain the following set of constraints, which are no more constraining than the previous set,
\begin{align*}
 &2^{-\wsep(3)-\Delta\wsep(3)-\Delta\wright(d_r)} \notag\\
+&2^{-2\wsep(3)-\wright(d_r)-d_r\cdot \degdec} \notag\\
\le& 1, & 2\le d_r\le 3.
\end{align*}
Here, the number of vertices in $N_{I-A}(N_{I-A}[s]) \cap (S\cup R)$ is $d_r$, accounting for the vertex $r'$ and the neighbors of $r$ besides $s$.

Now, consider the case where $s$ has two neighbors in $L$ or in $R$.
Since the graph is bipartite, these two neighbors are not adjacent.
If $\mu_r(R)-\mu_r(L)\le B$, then the two created subinstances are balanced, which gives us the following set of constraints:
\begin{align}
 &2^{-\wsep(3)-1/2\cdot(\Delta\wright(d_1)+\Delta\wright(d_2)+\Delta\wright(d_3))} \notag\\
+&2^{-\wsep(3)-1/2\cdot(\wright(d_1)+\wright(d_2)+\wright(d_3))-(d_1+d_2+d_3-3)\cdot \degdec} \notag\\
\le& 1, & 2\le d_1,d_2,d_3\le 3. \label{eq:branch2}
\end{align}
Otherwise, $\mu_r(R)-\mu_r(L) > B$.
If $s$ has two neighbors in $R$, then the worst case is the balanced one, which is covered by \eqref{eq:branch2}.
Finally, if $s$ has two neighbors in $L$, then its neighbor $r\in R$ has $d_{I-A}(r)=2$ due to line \ref{nl:imbal-2L-cond}.
The imbalanced case is not covered by \eqref{eq:branch2} and incurs the following constraint on the measure.
\begin{align}
 &2\cdot2^{-\wsep(3)-\wright(2)-\Delta\wright(d)} \notag\\
\le& 1, & 2\le d\le 3.
\end{align}
Note that the deletion of $s$ triggers that $r$ is removed by line \ref{nl:deg1} in Algorithm \ref{algo:cds}, which decreases the degree of its other neighbor, which is in $R$ due to line \ref{nl:imbal-2LS-cond}.

\medskip
This gives us all the constraints on $\mu_3$ for the analysis.
To obtain values for the various weights, we set
\begin{align}
 \welt(i) &= \welt(i+1), & \wset(i) &= \wset(i+1), & i\ge 6,
\end{align}
and we minimize $\wset(6)+\welt(6)$.
Since the separation has not added any non-convex constraints, the resulting mathematical program is convex, as usual \cite{hybrid}.
Solving this convex program gives the following values for the weights:
\begin{align*}
 \wright(2) &= 0.15282 & \welt(2) &= 0.15384 & \wset(2) &= 0.16408\\
 \wright(3) &= 0.22669 & \welt(3) &= 0.22732 & \wset(3) &= 0.24592\\
 \wsep(2)   &= 0.75630 & \welt(4) &= 0.26684 & \wset(4) &= 0.29320\\
 \wsep(3)   &= 0.78943 & \welt(5) &= 0.29023 & \wset(5) &= 0.30224\\
 &&\welt(6) &= 0.30019 & \wset(6) &= 0.30224\\
\end{align*}
For subcubic instances, the depth of the search tree is bounded by a polynomial since each recursive call decreases the measure $(|S_2|+|S|)\cdot \frac{\mu_r(V(I))}{\wright(2)} + \left| \frac{\mu_r(R)-\mu_r(L)}{\wright(2)} \right|$ by at least 1, where $S_2$ is the set of vertices in $S$ that have degree 2 in $I-A$.
By Lemma \ref{lem:measureanalysis}, we conclude that subcubic instances are solved in time $O^*(2^{\mu_3})$.
Now, using Lemma \ref{lem:combinemeasureanalysis} with $\mu = \mu_4$, $\mu' = \mu_3$ and $\eta = |V(I)\setminus A|$, we conclude that the running time of the algorithm is upper bounded by $\Ostar{2^{\mu_4}}$.
Thus, the algorithm solves \textsc{\#Dominating Set} in time $\Ostar{2^{(\welt(6)+\wset(6))\cdot n}} = \Ostar{1.5183^n}$, where $n$ is the number of vertices.

\begin{theorem}
 Algorithm \#SC solves \textsc{\#Set Cover} in
 time $\Ostar{2^{0.30019\cdot |\mathcal{U}|+0.30224 \cdot |\mathcal{S}|}}$ and
 \textsc{\#Dominating Set} in time $O(1.5183^n)$,
 using polynomial space.
\end{theorem}

\subsection{Comments on the \#SC algorithm and its analysis}

Let us make a few final comments on the algorithm and its analysis.
First, while the algorithm of van Rooij always prefers to branch on an element when the graph contains both an element and a set of maximum degree, this is not always possible when the algorithm needs to branch on vertices in the separator.
However, this disadvantage is overwhelmed by the gain due to the separator branching.

Second, a certain amount of care needs to be taken to make sure that the algorithm terminates.
For example, if we had omitted ``that has degree $3$ in $I-A$'' in line \ref{nl:imbal-2L-cond} of Algorithm \ref{algo:3cds}, then lines \ref{nl:deg2-S-bal} and \ref{nl:imbal-2L} could have eternally dragged vertices alternately to the left and to the right.

Third, the initial \textsc{Set Cover} transformation is not absolutely necessary, since we may equivalently label the vertices as in Subsection~\ref{subsec:3cds}.
We can then assign a weight to each vertex which depends on its label and the number of neighbors with each label.
However, our attempt resulted in an unmanageable number of constraints.\footnote{With 8Gb of memory, we were only able to handle around 1.25 million constraints using a 64-bit version of AMPL and the solver NPSOL, and such instances were solved in around 10 minutes. This was not even sufficient to analyze running times for degree-4 instances.}
We tried several compromises to simplify the analysis (merge labels, simplify the degree-spectrum), and we found that, of these, the simplification corresponding exactly to the \textsc{Set Cover} translation performed best.
Here, a graph vertex $u$ corresponds to an element $e$, which encodes that the vertex needs to be dominated, and to a set $X$, which encodes that it can be in the dominating set.
Thus, in the measure, the weight of $u$ is the weight of $e$ plus the weight of $X$.
Therefore, the weight of $u$ depends on the degree of $e$ and the degree of $X$, but not on the label of $u$ and the degree of $u$.
It would be interesting to know whether this choice of weights compromises the optimality of the analysis.

Finally, we remark that it did not help the analysis to allow different weights for sets and elements in the subcubic analysis. To further improve the running time, it will not be sufficient to improve the subcubic case, since none of the degree-3 constraints turn out tight.

\section{Limitations of the method}
\label{sec:limitations}

In this section, we discuss when we can and cannot expect to achieve improved running times using the \smc method.
The most important feature an algorithm needs to have is that it eventually encounters instances where small balanced separators can be computed efficiently.
This is the case for algorithms that branch on graphs of bounded degree in their final stages, but other means of obtaining graphs with relatively small separators are conceivable. This might include cases where the treewidth of the graph is bounded by a fraction of the number of vertices, where the instance has a small backdoor set \cite{GaspersS12} to bounded treewidth instances (as in \cite{GaspersS13} for incidence graphs of SAT formulae), and graphs with small treewidth modulators \cite{KimLPRRSS13,FominLMS12}.

The first limitation of our method occurs when a non-separator based algorithm is already so fast that branching on separators does not add an advantage.
For example, the current fastest algorithm for \textsc{Maximum Independent Set} on subcubic graphs has worst-case running time
$O(1.0836^n)$, where $n$ is the number of vertices of the input graph \cite{XiaoN13}.
Merely branching on all the vertices of the separator would take time $2^{n/6+o(n)}$ and $2^{1/6} \approx 1.1225 > 1.0836$.
This full time will be required, e.g., if the computed separator happens to be an independent set, so that
a decision for one vertex of the separator (put it in the independent set or not) does not have any direct implications for the other vertices in the separator.
While this is a limitation of our method, it also motivates the study of balanced separators with additional properties.
A small but relatively dense separator could be useful for \textsc{Maximum Independent Set} branching algorithms.

The second limitation is that our \smc subroutines can often be replaced by treewidth-based subroutines \cite{FominGSS09}, leading to smaller worst-case running times but exponential space usage.
In fact, replacing the branching on graphs with small maximum degree by a treewidth-based dynamic programming algorithm, such as \cite{RooijBR09}, is a well-known method \cite{FominGSS09} to improve the worst-case running time of an algorithm at the expense of exponential-space usage.
For instance, the currently fastest exponential-space algorithm for \textsc{\#Dominating Set} \cite{NederlofRD14} uses a treewidth-based subroutine on bounded degree instances to improve the running time.
It solves \textsc{\#Dominating Set} in $O(1.5002^n)$ time and exponential space.
Using only polynomial space is however recognized as a significant advantage,
especially for algorithms that may run much faster on real-world instances
than their worst-case running time bounds,
and the field devotes attention to both categories \cite{FominK10}.
An exponential space usage would quickly become a bottleneck in the execution of the algorithm.

\section{Conclusion}
\label{sec:concl}

We have presented a new method to analyze separator based branching algorithms within the \mc framework.
It uses a novel kind of measure that is able to take advantage of a global
structure in the instance and amortize a sudden large gain,
due to the instance decomposing into several independent subinstances,
over a linear number of previous branchings.

The method provides opportunities and challenges in the design
and analysis of other exponential-time branching algorithms: opportunities
because we believe it is widely applicable,
and challenges because it complicates the analysis and leads to more choices in the
design of algorithms.
As usual, finding a set of reductions leading to a small running time upper bound in a \mc analysis remains an art.
The method has also been used to obtain faster polynomial-space algorithms for counting independent sets and graph coloring \cite{GaspersL17}.

Finally, our analysis might open the way for other measures relying on global structures of instances.

\begin{acks}
	
The research was supported in part by the
DIMACS 2006--2010 Special Focus on Discrete Random Systems,
NSF grant DMS-0602942.
Serge Gaspers is the recipient of an Australian Research Council (ARC) Discovery Early Career Researcher Award (DE120101761) and a Future Fellowship (FT140100048), and acknowledges support under the ARC's Discovery Projects funding scheme (DP150101134).
The research was done in part at Dagstuhl Seminars 10441
(Exact Complexity of NP-hard problems, 2010)
and 13331
(Exponential Algorithms: Algorithms and Complexity Beyond Polynomial Time, 2013),
and was presented at the latter.
\end{acks}

\bibliographystyle{ACM-Reference-Format}
\bibliography{separate-long,separate}

\end{document}